%% file: main.tex
\documentclass[twoside]{article}

\usepackage[accepted]{aistats2025}
\usepackage{algorithm}
\usepackage{algorithmic}
\usepackage{amsmath}
\usepackage{amssymb}
\usepackage{mathtools}
\usepackage{amsthm}
\usepackage{booktabs}
\usepackage{xcolor}
\usepackage{subcaption}

\usepackage[round]{natbib}

\input{def}

\begin{document}

%

%

\twocolumn[

\aistatstitle{On Distributional Discrepancy for Experimental Design with General Assignment Probabilities}

\aistatsauthor{Anup B. Rao \And Peng Zhang}

\aistatsaddress{ Adobe Research \And  Rutgers University} 

\vspace{-.8cm}
\aistatsaddress{anuprao@adobe.com \And pz149@rutgers.edu } 
]

\begin{abstract}

We investigate experimental design for randomized controlled trials (RCTs) with both equal and unequal treatment-control assignment probabilities. Our work makes progress on the connection between the distributional discrepancy minimization (DDM) problem introduced by \cite{HSSZ19} and the design of RCTs. 
We make two main contributions: First, we prove that approximating the optimal solution of the DDM problem within a certain constant error is NP-hard. Second, we introduce a new Multiplicative Weights Update (MWU) algorithm for the DDM problem, which improves the Gram-Schmidt walk algorithm used by \cite{HSSZ19} when assignment probabilities are unequal. 
Building on the framework of \cite{HSSZ19} and our MWU algorithm, we then develop the MWU design, which reduces the worst-case mean-squared error in estimating the average treatment effect. 
Finally, we present a comprehensive simulation study comparing our design with commonly used designs.

\end{abstract}

\section{INTRODUCTION}

Randomized Controlled Trials (RCTs) are the ``gold standard'' for estimating the causal effects of a new treatment \citep{HR10,Morgan_Winship_2014,IR15}.
In an RCT, experimental units are randomly assigned into one of two groups: a treatment group, which receives the new treatment, and a control group, which receives the standard treatment. 
The outcomes of these groups will be compared to estimate the causal effects of the new treatment.
Ideally, the two groups are ``similar", so the only difference is the treatment they receive.
The \emph{design} of an RCT refers to the distribution of the random assignment.
It involves a trade-off between balancing observed covariates and being robust to unobserved confounders and model misspecification.
This trade-off was first introduced by \cite{efron1971forcing} and has been central in commonly used designs, 
such as randomized blocking, pairwise matching, and rerandomization.

A recent breakthrough by \cite{HSSZ19} introduced the \emph{Distributional Discrepancy Minimization (DDM)} problem, offering a precise mathematical framework for balancing covariates while preserving robustness. 
This approach achieved a nearly optimal trade-off between balance and robustness, leading to more accurate causal effect estimates. 
Their framework has since inspired further advancements in RCT design \citep{ADMR22,chatterjee2023central}.

However, the GSW design proposed in \cite{HSSZ19} has a notable limitation.  As we show in Section~\ref{sec:problem_set}, the GSW design does not provide nearly optimal guarantees when the probabilities of treatment and control assignments are \emph{unequal} (i.e., not $50/50$).
We address this gap in this paper.

Unequal assignment probabilities are important in RCTs to reduce costs, address ethical concerns, and manage differences in response variances between groups \citep{torgerson1997unequal,dumville2006use, wong2008optimum, ryeznik2018comparative,sverdlov2019implementing,azriel2022optimality}. 
However, few studies have focused on balancing covariates in this context. For example, the authors of \cite{azriel2022optimality} state, ``We are not aware of work that discusses unequal-allocation design vis-a-vis the consideration of minimizing observed [covariate] imbalance in the non-sequential setting where all $\xx$’s [covariates] are known a prior." 
Our results expand on this understanding for unequal assignment probabilities.

\subsection{Our Contributions}

This paper builds on the framework of \cite{HSSZ19} and makes progress in solving the DDM problem for \emph{unequal} assignment probabilities, further leading to new contributions in the design of RCTs.

First, we prove that achieving (nearly) \emph{instance-optimality} for the DDM problem is NP-hard. 
Assuming P$\neq$NP, \emph{no} polynomial-time algorithm can guarantee that, for every instance, it returns an assignment that approximates the optimal assignment within a certain constant error, even if the optimal solution perfectly balances the covariates.

Second, we present a new Multiplicative Weights Update (MWU) algorithm that finds improved solutions for the DDM problem for the \emph{full} spectrum of assignment probabilities. 
Theoretically, we show that our algorithm improves the GSW algorithm when assignment probabilities are unequal and matches its performance when probabilities are equal, both up to constant factors.
Empirically, we demonstrate that our algorithm consistently produces solutions to the DDM problem with the lowest objective values across \emph{all} assignment probabilities, outperforming the GSW algorithm and other commonly used designs on both synthetic and real-world datasets.
While the MWU framework has been widely used in optimization and machine learning, its application to the DDM problem and the design of RCTs is new. 

Third, we propose the MWU design building on the framework of \cite{HSSZ19} and our MWU algorithm.
Our design enhances the GSW design on the trade-off between covariate balance and robustness when assignment probabilities are unequal, and further reduces the mean-squared error (MSE) in estimating the average treatment effect (ATE). In addition, our simulation results show that our design achieves lower MSE in estimating the ATE when outcomes depend linearly or nearly linearly on covariates, compared with both the GSW design and other commonly used designs.

\subsection{Related Work}

\paragraph{Distributional Discrepancy Minimization and the Design of RCTs.}
Our work builds upon the rigorous framework of analyzing the balance-robustness trade-off provided in ~\cite{HSSZ19}. A tighter asymptotic analysis of the GSW design introduced in that paper was given in \cite{chatterjee2023central}. The framework was adapted to online design in ~\cite{ADMR22}. 
The GSW design was shown to have an optimal trade-off between balance and robustness when treatment-control assignment probabilities are \emph{equal}. While the design can adapt to unequal probabilities, it can be \emph{sub-optimal}, as discussed in Section \ref{sec:prob_set_subopt}. 

\paragraph{Balance-Robustness Trade-off and Other Commonly Used Designs.} 
There are various designs that span the spectrum between covariate balance and robustness.
On one end of the spectrum is the Bernoulli design and the Complete Randomization. They uniformly sample an assignment from all assignments that satisfy marginal probability conditions or group-size conditions regardless of covariates.
They have the strongest robustness \citep{kallus2018optimal,azriel2022optimality,HSSZ19}, but can cause covariate imbalance by chance.
On the other end is the optimal designs, first suggested by \cite{student1938comparison} and then expanded by \cite{bertsimas2015power}, \cite{kasy2016experimenters}, \cite{deaton2018understanding}, \cite{kallus2018optimal}, and \cite{bhat2020near}.
They define a measure of covariate balance and find the best possible assignment that minimizes covariate imbalance using tools from numerical or combinatorial optimization. 
The best assignment is usually deterministic and thus may lack robustness.

Various designs lie between these two extremes, trading off some robustness to ensure covariates balance is considered important.
Pairwise matching designs pair units based on covariate similarity and then randomized within each pair \citep{greevy2004optimal,imai2009essential,bai2022inference}. 
Randomized block designs group units with similar covariates into blocks and then randomize within each block \citep{Fisher35,higgins2016improving,azriel2022optimality}. 
Rerandomization repeatedly generates random assignments uniformly from all feasible assignments until one meets a pre-specified covariate balance criterion, at which point it is accepted \citep{MR12,li2018asymptotic,li2020rerandomization}.

Pairwise matching cannot be used for unequal assignment probabilities since it assigns exactly one unit from each matched pair to treatment and one to control. 
Randomized block design and rerandomization can adapt unequal assignment probabilities. However, they do not perform well when the number of covariates is large \citep{branson2021ridge,zhang2024pca,davezies2024revisiting}, and they do not provide a formal analysis of the balance-robustness trade-off \citep{HSSZ19}. 

\paragraph{Discrepancy Theory.}

The distributional discrepancy minimization problem introduced in \cite{HSSZ19} is closely related to discrepancy theory, a subfield of discrete mathematics and theoretical computer science \citep{matousek99,chazelle01,chen2014panorama}. 
The GSW design in \cite{HSSZ19} builds on the GSW algorithm developed by \cite{BDGL18} within the context of discrepancy theory.
Besides \cite{HSSZ19} and \cite{BDGL18}, other algorithms in discrepancy theory have also inspired RCT design, either directly or indirectly \citep{krieger2019nearly,turner2020balancing}.
Our work continues this line of research by further connecting algorithmic discrepancy theory with the design of RCTs.
Finally, recent advancements in algorithmic discrepancy theory potentially suggest new improvements in the design of RCTs (for example, \cite{bansal10}, \cite{lovettM15}, \cite{rothvoss14}, \cite{eldan2018efficient}, \cite{levy2017deterministic}, \cite{BDGL18}, \cite{ALS21}, \cite{bansal2022unified}, \cite{pesenti2023discrepancy}, \cite{KRR23} and the references therein).

\paragraph{Roadmap.} 
We introduce notations in Section \ref{sec:notation}. We outline the problem setting and formally define the DDM problem in Section \ref{sec:problem_set}.
We then formally state our contributions in Section \ref{sec:contribution}, and present our MWU algorithm in Section \ref{sec:algo}.
Finally, we provide a comprehensive empirical study in Section \ref{sec:experiment}. 
Due to space constraints, all proofs are deferred to our Supplementary Material.

\section{NOTATIONS}
\label{sec:notation}

In this paper, we use {\bf bold} letters for vectors and matrices and regular letters for scalars.
For $n \in \mathbb{N}$, we let $[n] = \{1,\ldots,n\}$.
For a vector $\xx \in \mathbb{R}^n$, let $\norm{\xx}$ be its Euclidean norm. 
For a matrix $\AA \in \mathbb{R}^{n \times n}$, let $\norm{\AA}$ be its operator norm induced by Euclidean norm, defined as $\norm{\AA} = \sup_{\norm{\xx}=1} \norm{\AA \xx}$.
We let $\tr(\AA)$ be the trace of $\AA$.
In addition, we define the norm $\norm{\AA}_{1,2}$ as the maximum Euclidean norm among $\AA$'s columns. 
For two vectors $\xx, \yy \in \mathbb{R}^n$ or two matrices $\AA, \BB \in \mathbb{R}^{n \times n}$, we let $\langle \xx, \yy \rangle = \xx^\top \yy$ and $\langle \AA, \BB \rangle = \tr(\AA \BB^\top)$ denote their inner products respectively, where $\xx^\top$ and $\BB^\top$ are the transposes of $\xx$ and $\BB$ respectively. 
We let $\one$ be the all-one vector, and $\II$ be the identity matrix.
For a random vector $\yy$ drawn from distribution $\mathcal{D}$, let $\cov_{\yy \sim \mathcal{D}}(\yy)$ denote the covariance matrix of $\yy$, defined as the expected value of the outer product of $\yy - \mathbb{E}_{\yy \sim \mathcal{D}}[\yy]$ with itself.
When the context is clear, we drop the subscript $\yy \sim \mathcal{D}$ from the covariance matrix and expectation.

\input{setup}
\input{results}

\input{algo}

\input{experiment}

\section{CONCLUDING REMARKS}
\label{sec:future}

We present a new MWU algorithm for the distributional discrepancy minimization problem introduced by \cite{HSSZ19}. Our algorithm outperforms the GSW design for unequal assignment probabilities, which have important applications in experimental design. Building on the framework of \cite{HSSZ19}, our approach reduces the mean-squared error in estimating the ATE compared to commonly used designs, strengthening the connection between distributional discrepancy and experimental design.

One limitation of our design is that we assume the relationship between covariates and outcomes is nearly linear, the same as \cite{HSSZ19} and many others. As suggested by \cite{HSSZ19}, this limitation may be addressed by incorporating higher-order covariate terms and their interactions or using kernel methods. 

Our algorithm is slower than the GSW design and other designs, but its runtime remains \emph{polynomial} in input size, which is acceptable for many small- and moderate-sized randomized experiments in fields like medicine, agriculture, and education. The higher computational cost of planning could be outweighed by the increase in estimation accuracy.

\subsubsection*{Acknowledgments}

PZ acknowledges the Office of Advanced Research Computing (OARC) at Rutgers, The State University of New Jersey, for providing access to the Amarel cluster and associated research computing resources that have contributed to the results reported here. URL: https://it.rutgers.edu/oarc.
PZ was supported by NSF Grant CCF-2238682, an Adobe Data Science Research Award, and a Rutgers Research Council Individual Fulcrum Award.


\bibliography{ref}
\bibliographystyle{abbrvnat}
\input{my_supplement}

\end{document}

%% file: def.tex
\theoremstyle{plain}
\newtheorem{theorem}{Theorem}[section]
\newtheorem{proposition}[theorem]{Proposition}
\newtheorem{lemma}[theorem]{Lemma}

\theoremstyle{definition}

\theoremstyle{remark}

\newtheorem{problem}[theorem]{Problem}
\newtheorem{claim}[theorem]{Claim}

\def\pleq{\preccurlyeq}
\def\pgeq{\succcurlyeq}

\def\var#1{\mbox{\bf Var}\left[ #1 \right]}

\def\defeq{\stackrel{\mathrm{def}}{=}}

\def\abs#1{\left|#1  \right|}

\def\norm#1{\left\| #1 \right\|}

\newcommand{\one}{\mathbf{1}}

\def\calE{\mathcal{E}}

\def\calS{\mathcal{S}}

\def\eps{\epsilon}

\newcommand\PPi{\boldsymbol{\Pi}}

\newcommand\bbeta{\boldsymbol{\beta}}

\newcommand\ddelta{\boldsymbol{\delta}}

\renewcommand\aa{\boldsymbol{\mathit{a}}}
\newcommand\bb{\boldsymbol{\mathit{b}}}

\newcommand\ff{\boldsymbol{\mathit{f}}}
\newcommand\pp{\boldsymbol{\mathit{p}}}

\newcommand\uu{\boldsymbol{\mathit{u}}}

\newcommand\ww{\boldsymbol{\mathit{w}}}
\newcommand\yy{\boldsymbol{\mathit{y}}}
\newcommand\zz{\boldsymbol{\mathit{z}}}
\newcommand\xx{\boldsymbol{\mathit{x}}}

\newcommand\mmu{\boldsymbol{\mathit{\mu}}}

\renewcommand\AA{\boldsymbol{\mathit{A}}}
\newcommand\BB{\boldsymbol{\mathit{B}}}
\newcommand\CC{\boldsymbol{\mathit{C}}}
\newcommand\DD{\boldsymbol{\mathit{D}}}

\newcommand\II{\boldsymbol{\mathit{I}}}

\newcommand\MM{\boldsymbol{\mathit{M}}}

\newcommand\WW{\boldsymbol{\mathit{W}}}
\newcommand\VV{\boldsymbol{\mathit{V}}}
\newcommand\XX{\boldsymbol{\mathit{X}}}
\newcommand\YY{\boldsymbol{\mathit{Y}}}

\newcommand{\diag}{\textsc{diag}}

\DeclareMathOperator*{\argmin}{arg\,min}
\DeclareMathOperator*{\argmax}{arg\,max}

\def\tr{\text{tr}}
\def\var{\text{Var}}

\def\cov{\text{Cov}}
\def\ssplit{\text{Set-Splitting}}

\def\dims{d}

%% file: setup.tex
\section{PROBLEM SETUP}
\label{sec:problem_set}

In this section, we present the assumptions of RCTs, covariate balance and robustness, and the Distributional Discrepancy Minimization (DDM) problem introduced by \cite{HSSZ19}.

We follow the Neyman-Rubin potential outcome framework \citep{rubin2005causal} for an RCT with a population of $n$ units and two treatment groups: treatment and control. Each unit $i \in [n]$ has two potential outcomes: $a_i$ under treatment and $b_i$ under control. If unit $i$ is assigned to the treatment group, we observe the outcome $a_i$; otherwise, we observe the outcome $b_i$.

The experimenter needs to randomly assign each unit $i$ to either the treatment or control group, with respective pre-specified marginal probabilities $p_i$ and $1-p_i$. 
Let $\zz = (z_1, \ldots, z_n)^\top \in \{\pm 1\}^n$ represent the random assignment of the $n$ units, where $z_i = 1$ indicates that unit $i$ is assigned to the treatment group, and $z_i = -1$ indicates assignment to the control group. 
In this paper, we restrict ourselves to \emph{feasible designs/assignments} in which, for each $i \in [n]$, $\Pr(z_i = 1) = p_i$ and $\Pr(z_i = -1) = 1 - p_i$.

We assume that the potential outcomes of each unit are deterministic and the only source of randomness comes from the random assignment of the units. 
This model is known as the \emph{randomization model} or \emph{Neyman model} \citep{Fisher35,kempthorne1955randomization,rosenberger2015randomization}.

We want to estimate the \emph{average treatment effect (ATE)}: \begin{align}
\tau \defeq \frac{1}{n} \sum_{i=1}^n (a_i - b_i).
\label{eqn:ate}
\end{align}
We will use the \emph{Horvitz-Thompson (HT) estimator}: 
\begin{align}
\hat{\tau} \defeq \frac{1}{n} \left( \sum_{i: z_i=1} \frac{a_i}{p_i} - \sum_{i: z_i=-1} \frac{b_i}{1-p_i} \right).
\label{eqn:ht_estimator}
\end{align}
The HT estimator $\hat{\tau}$ is \emph{unbiased} for a feasible assignment, meaning that $\mathbb{E}[\hat{\tau}] = \tau$.
We want to minimize the \emph{mean-squared error (MSE)} of $\hat{\tau}$, defined as
\begin{align}
\text{MSE}_{\zz}(\hat{\tau}) \defeq 
\mathbb{E}_{\zz}[(\hat{\tau} - \tau)^2] = \frac{1}{n^2} \mmu^\top \cov(\zz) \mmu,
\label{eqn:mse_def}
\end{align}
where $\mmu = (\mu_1, \ldots, \mu_n)$ and each $\mu_i = \frac{a_i}{p_i} + \frac{b_i}{1-p_i}$ is a weighted sum of the two potential outcomes for $i \in [n]$. 
The vector $\mmu$ is called the \emph{potential outcome vector}.

Since $\mmu$ is fixed but unknown, we want to find a feasible design that minimizes the \emph{worst-case} MSE among all $\mmu$ (up to scaling).
This ensures that our estimation is robust even with an adversary that provides the worst possible $\mmu$. 
The worst-case MSE has been studied by \cite{efron1971forcing}, \cite{kallus2018optimal}, \cite{kapelner2021harmonizing}, \cite{HSSZ19}, and others.

\subsection{The Distributional Discrepancy Minimization Problem}
\label{sec:problem_set_opt}

We assume that each unit $i \in [n]$ has $d$ covariates -- pre-treatment variables observed before the trial, denoted by $\xx_i \in \mathbb{R}^d$. We also define $\XX = (\xx_1, \cdots, \xx_n)^\top \in \mathbb{R}^{n \times d}$. 
We can normalize the covariate vectors so that $\norm{\xx_i} \le 1$ for every $i \in [n]$ (for example, by dividing each $\xx_i$ with $\max_i \norm{\xx_i}$).

We can write the potential outcome vector $\mmu = \XX \bbeta + \ddelta$ where $\bbeta = \argmin_{\bbeta'} \norm{\mmu - \XX \bbeta'}$ and $\ddelta$ is orthogonal to the columns of $\XX$. 
Both $\bbeta \in \mathbb{R}^d$ and $\ddelta \in \mathbb{R}^n$ are assumed to be fixed but unknown.
Using this decomposition of $\mmu$, we can write the MSE as
\begin{align*}
\text{MSE}_{\zz}(\hat{\tau}) 
& = \frac{1}{n^2} \left( \bbeta^\top \cov(\XX^\top \zz) \bbeta + \ddelta^\top \cov(\zz) \ddelta 
 \right. \\ & ~~~~~~~~~~~ \left.
+ 2 \bbeta^\top \cov(\XX^\top \zz) \ddelta
\right) \\
& \le \frac{2}{n^2} \left( \| \cov(\XX^\top \zz) \| \norm{\bbeta}^2 + \|\cov(\zz) \| \norm{\ddelta}^2
\right).
\end{align*}
Also see \cite{kapelner2021harmonizing} and \cite{HSSZ19}. 
In the above equation, $\| \cov(\XX^\top \zz) \|$ measures covariate balance, and $\| \cov(\zz) \|$ captures robustness against unobserved variables or model misspecification.

We build on the framework of \cite{HSSZ19} to simultaneously balance covariates and maintain robustness. 
The GSW design in \cite{HSSZ19} has a \emph{design parameter} $\phi \in [0,1]$, chosen by the experimenter, to govern the trade-off between covariate balance and robustness. 
One constructs an \emph{augmented covariate matrix} $\BB$: if $\phi \in (0,1)$, let
\begin{align}
\BB = \begin{pmatrix}
    \sqrt{\phi} \II \\
    \sqrt{1-\phi} \XX^\top
\end{pmatrix} \in \mathbb{R}^{(n+d) \times n};
\label{eqn:mat_B}
\end{align}
if $\phi = 1$, simply let $\BB = \II$ (only robustness) and if $\phi = 0$, let $\BB = \XX^\top$ (only covariate balance).
A smaller $\phi$ emphasizes more on covariate balance.
The authors reduced finding feasible $\zz$ that minimizes the worst-case MSE to the following problem.

\begin{problem}[The Distributional Discrepancy Minimization (DDM) Problem]
Given $\BB \in \mathbb{R}^{m \times n}$ with $\|\BB\|_{1,2} \le 1$ and $\pp = (p_1, \ldots, p_n) \in (0,1)^n$, find a random vector $\zz \in \{\pm 1\}^n$ sampled from
\begin{align}
\mathcal{D}^* \in 
\argmin_{\mathcal{D} \text{ is feasible}} f_{\BB}(\mathcal{D}) \defeq \| \cov_{\zz \sim \mathcal{D}} (\BB \zz) \|.
\label{eqn:goal}
\end{align}
In addition, we want to develop a computationally efficient (i.e., polynomial-time) algorithm that returns such a $\zz$.
\label{prob}
\end{problem}

The GSW design finds a feasible $\zz$ such that $\| \cov(\BB \zz) \| \le 1$ for any $\BB$ satisfying the conditions in Problem \ref{prob}. 
It implies that, for any design parameter $\phi \in (0,1)$, $\|\cov(\XX^\top \zz) \| \le \frac{1}{1-\phi}$ and $\|\cov(\zz) \| \le \frac{1}{\phi}$; for $\phi = 1$, $\|\cov(\zz) \| \le 1$, and for $\phi = 0$, $\|\cov(\XX^\top \zz) \| \le 1$.


\subsection{Sub-Optimality of the GSW Design for Unequal Probabilities}
\label{sec:prob_set_subopt}

The GSW design's guarantee $\| \cov(\BB \zz) \| \le 1$ is optimal when \(\pp = (1/2) \one\) (i.e., equal treatment-control assignment probabilities).
However, this guarantee can be \emph{far from optimal} when \(\pp \neq (1/2) \one\) (i.e., unequal probabilities). 
For example, consider $\| \XX^\top \XX \| \le 10$ and \(\pp = (0.01)\one\). 
If we use the Bernoulli design, assigning each unit to treatment with a probability of \(0.01\) and to control with \(0.99\), then \(\cov(\zz_{Bernoulli})\) is a diagonal matrix with diagonal entries equal to \(0.0396\), resulting in \(\|\cov(\zz_{Bernoulli})\| = 0.0396\).
In addition, $\|\cov(\XX^\top \zz_{Bernoulli})\| \le \|\XX^\top \XX\| \cdot \| \cov(\zz_{Bernoulli})  \| \leq 0.396$. 
These values are significantly smaller than the upper bound of $1$ provided by the GSW design, regardless of the choice of \(\phi\).
In addition, in Supplementary Section \ref{sec:numeric_instance}, we provide a numerical example where \(\|\cov(\BB \zz_{Bernoulli})\| < \|\cov(\BB \zz_{GSW}) \|\) for a specific instance of the DDM problem, and then we present an artificial example showing that it is possible for \(\|\cov(\BB \zz_{Bernoulli})\| = O(n^{-1}) \cdot \|\cov(\BB \zz_{GSW}) \|\).
To the best of our knowledge, we are the first to establish improved bounds for the DDM problem with unequal assignment probabilities.

%% file: results.tex
\section{OUR RESULTS}
\label{sec:contribution}

\subsection{Hardness Results}

We establish a strong NP-hard result showing that, assuming P$\neq$NP, we \emph{cannot} approximate the optimum of the DDM problem, described in Problem \ref{prob}, within a certain constant additive error, even if the optimum is $0$.

\begin{theorem}
There exists a universal constant $c > 0$ such that the following holds:
For any $n \in \mathbb{N}$ and $\alpha \in (0,1/2), \beta \in (0,1)$, there exists $\pp \in \{1-\alpha, 1-2\alpha(1-\beta), 1-2\alpha \beta \}^n$ such that 
it is NP-hard to distinguish between the following two cases of $\BB \in \mathbb{R}^{\Theta(n) \times n}$ with $\|\BB \|_{1,2} \le 1$: (1) $f_{\BB}(\mathcal{D}^*) = 0$ and (2) $f_{\BB}(\mathcal{D}^*) \ge c \alpha^2 (2\beta-1)^2$. Therefore, 
\emph{no} polynomial-time algorithm can, for any $\BB \in \mathbb{R}^{\Theta(n) \times n}$ with $\|\BB \|_{1,2} \le 1$, return a feasible random $\zz \in \{\pm 1\}^n$ whose distribution $\mathcal{D}$ satisfies $f_{\BB}(\mathcal{D}) \le f_{\BB}(\mathcal{D}^*) + c \alpha^2 (2\beta-1)^2$, unless P$=$NP.
\label{thm:main2}
\end{theorem}

The parameters $\alpha, \beta$ in Theorem \ref{thm:main2} can depend on the dimensions of $\BB$.
When $\alpha, \beta$ are both constants and $\beta \neq 1/2$, the theorem means there is a \emph{constant} gap between any $f_{\BB}(\mathcal{D})$ achievable in polynomial time and the optimal $f_{\BB}(\mathcal{D}^*)$.
When $\beta$ is a fixed constant different from $1/2$ and $\alpha$ approaches $0$, the assignment probabilities (i.e., entries of $\pp$) approach either $1$ or $0$, and the gap between $f_{\BB}(\mathcal{D})$ achievable in polynomial time and $f_{\BB}(\mathcal{D}^*)$ goes to $0$.
For the equal probabilities case, that is, $\pp=(1/2)\one$, we establish a similar NP-hard result.
We defer the theorem statement to Supplementary Section \ref{sec:appendix_hard}.

Our proofs build on reductions from the $2$-$2$ \textsc{Set Splitting} problem \citep{guruswami04,charikar2005clustering,SZ22}. 
We defer the proofs to Supplementary Section \ref{sec:appendix_hard}.

\subsection{An Efficient MWU Design}

We develop a computationally efficient algorithm for the DDM problem that achieves the best performance of the GSW and Bernoulli designs for unequal assignment probabilities. 
Let \(\mathcal{D}_{GSW}\) and \(\mathcal{D}_{Bernoulli}\) denote the feasible distributions under the GSW and Bernoulli designs, respectively. For the function \(f_{\BB}\) defined in Equation \eqref{eqn:goal} of the DDM problem, we have:
\[
f_{\BB}(\mathcal{D}_{GSW}) \le 1, ~ f_{\BB}(\mathcal{D}_{Bernoulli}) = \| \BB \DD_{\pp} \BB^\top \|,
\]
where \(\DD_{\pp}\) is a diagonal matrix with diagonals \(4 p_i(1- p_i)\) for \(i \in [n]\).

\begin{theorem}
Given a matrix $\BB \in \mathbb{R}^{m \times n}$ with $\|\BB\|_{1,2} \le 1$ and a vector $\pp \in (0,1)^n$,
for any $\eps \in (0,1)$, 
we can find a random $\zz \in \{\pm 1\}^n$ drawn from a feasible distribution $\mathcal{D}$ such that
\begin{align}
f_{\BB}(\mathcal{D}) 
\le (1+\eps)^2 \min \left\{f_{\BB}(\mathcal{D}_{Bernoulli}), 1 + \frac{1}{\eps} \right\},
\label{eqn:thm_algo_upper}
\end{align}
and the runtime is polynomial in $m,n, \eps^{-1}$.
The algorithm that achieves the above guarantee is presented in Algorithm \ref{alg:mwu} MWU.
\label{thm:algo}
\end{theorem}

Consider $\eps$ being a constant independent of $m$ and $n$.  
When \(\pp\) deviates significantly from \((1/2) \one\), we have \( f_{\BB}(\mathcal{D}_{Bernoulli}) \ll 1 \), and Theorem \ref{thm:main2} ensures that \( f_{\BB}(\mathcal{D}_{MWU}) \le O(f_{\BB}(\mathcal{D}_{Bernoulli})) \ll 1 \), improving the best guarantee of GSW.  
When \(\pp \approx (1/2) \one\), Theorem \ref{thm:algo} establishes that \( f_{\BB}(\mathcal{D}_{MWU}) \) remains within a constant factor of the best guarantee of GSW.

Our theoretical upper bound in Equation \eqref{eqn:thm_algo_upper} may be conservative. To assess its practical performance, we conduct empirical experiments comparing our algorithm with the GSW algorithm, the Bernoulli design, and other commonly used designs. Our simulation results show that \( f_{\BB}(\mathcal{D}_{MWU}) \) consistently achieves the lowest values across \emph{all} assignment probabilities in both synthetic and real-world datasets.  
A promising direction for future work is to refine the upper bound in Equation \eqref{eqn:thm_algo_upper} and establish that \( f_{\BB}(\mathcal{D}_{MWU}) \) is smaller than \( O(f(\mathcal{D}_{GSW})) \) rather than the GSW upper bound, as suggested by our empirical findings.

Our algorithm differs from the SDP relaxation \citep{bhat2020near} and the generalized power method \citep{lu2022synthetic}, which minimize the MSE under different models. 
Since the DDM problem aims to find a \emph{distribution} over all feasible assignments rather than a \emph{single} assignment to minimize the MSE, we can develop a polynomial-time algorithm with stronger theoretical guarantees.

\paragraph{Estimating the ATE.}

Building on the framework of \cite{HSSZ19} and our MWU algorithm, we propose the MWU design. 
In this design, the experimenter selects a design parameter \(\phi \in [0,1]\) and an accuracy parameter \(\eps \in (0,1)\). The design then constructs an augmented matrix \(\BB\) as specified in Equation \eqref{eqn:mat_B} and runs the MWU algorithm (Algorithm \ref{alg:mwu}) on \(\BB\) with parameter $\eps$ to solve the DDM problem, returning a feasible assignment \(\zz\).

We obtain similar results on the balance-robustness trade-off, the expectation, variance, and convergence rate of the error of estimating the ATE under the MWU design, substituting the GSW upper bound with the MWU upper bound. 
An algorithm that achieves a better bound for the DDM problem immediately improves the estimation accuracy.
Our proofs are similar to those from \cite{HSSZ19}, and we include them in Supplementary Section \ref{sec:statistics} for completeness.

\begin{proposition}[Balance-robustness trade-off]
Suppose all covariate vector $\norm{\xx_i} \le 1$ after standard scaling.
Let $\phi \in (0,1)$ be the design parameter.
Let $\alpha_{MWU}$ be the upper bound of the MWU algorithm, as stated by the right-hand side of Equation \eqref{eqn:thm_algo_upper} in Theorem \ref{thm:algo}.
Let $\zz$ be the assignment returned by the MWU design.
Then, 
\[
\|\cov(\XX^\top \zz) \| \le \frac{\alpha_{MWU}}{1-\phi}, ~
\|\cov( \zz) \| \le \frac{\alpha_{MWU}}{\phi}.
\]
\label{prop:trade_off}
\end{proposition}

\begin{proposition}
The HT estimator for the ATE under the MWU design is unbiased, that is, $\mathbb{E}[\hat{\tau}] = \tau$.
\label{thm:unbias}
\end{proposition}

\begin{proposition}
Assume the conditions in Proposition \ref{prop:trade_off} hold.
Let $\mmu = (\mu_1, \ldots, \mu_n) \in \mathbb{R}^n$, where $\mu_i = \frac{a_i}{p_i} + \frac{b_i}{1-p_i}$ for each $i \in [n]$, be the potential outcome vector. Then, under the MWU design, the variance
\begin{align*}
& n \var(\hat{\tau}) \le \\
& \alpha_{MWU} \cdot \min_{\bbeta \in \mathbb{R}^d} \left\{ \frac{1}{\phi n} \|\mmu - \XX \bbeta \|^2 + \frac{1}{(1-\phi)n} \norm{\bbeta}^2 \right\}.
\end{align*}
\label{prop:var}
\end{proposition}

\begin{proposition}
Let $c_1, c_2 \in (0,1)$ and $c_3 > 0$ be fixed constants.
Assume that (1) the design parameter $\phi \ge c_1$, (2) every assignment probability $c_2 \le p_i \le 1-c_2$,
and (3) $\norm{\aa} \le c_3 \sqrt{n}, \norm{\bb} \le c_3 \sqrt{n}$. Then, under the MWU design, $\hat{\tau} - \tau \rightarrow 0$ in probability. Furthermore, the convergence rate satisfies $\hat{\tau} - \tau = \mathcal{O}_p (n^{-1/2})$. 
\label{prop:converge}
\end{proposition}

Propositions \ref{prop:trade_off} and \ref{prop:var} improve the GSW design when assignment probability $\pp$ deviates significantly from $(1/2)\one$, in which case we have $\alpha_{MWU} < 1$.
Theorem \ref{thm:algo} does not guarantee $\alpha_{MWU} < 1$ when $\pp \approx (1/2) \one$. 
However, our detailed empirical studies in Section \ref{sec:experiment} demonstrate that $\alpha_{MWU}$ is often much smaller than $1$ in practical scenarios. Additionally, the MWU design reduces the variance in estimating the ATE when outcomes and covariates are nearly linearly correlated, outperforming the GSW and many commonly used designs.

%% file: algo.tex
\section{OUR ALGORITHM}
\label{sec:algo}

In this section, we describe Algorithm \ref{alg:mwu} MWU for the DDM problem that achieves Theorem \ref{thm:algo}.
The algorithm is based on the matrix Multiplicative Weights Update method (MWU), which is commonly used in machine learning, optimization, and game theory \citep{arora2012multiplicative}.

A key idea of Algorithm \ref{alg:mwu} is to transform the problem of minimizing \(f_{\BB}(\mathcal{D}) = \|\cov_{\zz \sim \mathcal{D}}(\BB \zz) \|\) into a sequence of simpler tasks that minimize ``projections" of the covariance matrix onto positive definite (PD) matrices. Let \(\mathbb{S}_{++}^m\) denote the set of all symmetric PD matrices of dimensions \(m \times m\). The objective of the DDM problem can be rephrased as minimizing:
\[
f_{\BB}(\mathcal{D}) = \max_{\WW \in \mathbb{S}_{++}^{m}: \tr(\WW) = 1} \langle \cov_{\zz \sim \mathcal{D}}(\BB \zz), \WW \rangle.
\]
Algorithm \ref{alg:mwu} reduces this minimax problem into a sequence of subproblems with fixed \(\WW\). We represent a distribution \(\mathcal{D}\) using its support set \(Z\) and the probabilities associated with the vectors in \(Z\), which we will iteratively build. During this process, we update the weight matrix \(\WW\), which indicates the directions in which $\cov_{\zz \sim \mathcal{D}} (\BB \zz)$ needs improvement.
At each iteration, we find a feasible random vector \(\zz' \in \{\pm 1\}^n\) that minimizes \(\langle \cov_{\zz'}(\BB \zz'), \WW \rangle\). 
We then add \(\zz'\) to the set \(Z\) and adjust \(\WW\) based on the new set \(Z\). 
After enough iterations, the algorithm produces a feasible distribution \(\mathcal{D}\) supported on \(Z\) with a small value of \(f_{\BB}(\mathcal{D})\).

\begin{theorem}
Suppose we can access an oracle $\mathcal{O}(\BB, \WW, \pp)$ which takes a matrix $\BB \in \mathbb{R}^{m \times n}$ with $\|\BB\|_{1,2} \le 1$, 
a positive definite matrix $\WW \in \mathbb{R}^{m \times m}$, and a vector $\pp \in (0,1)^n$ as input and returns a random feasible vector $\zz' \in \{\pm 1\}^n$ whose distribution $\mathcal{D}'$ satisfies 
\begin{align}
\langle \cov_{\zz' \sim \mathcal{D}'}(\BB \zz'), \WW \rangle
\le \eta \cdot \tr(\WW).
\label{eqn:thm_mwu_oracle_condition}    
\end{align}
Then, for any $\eps \in (0,1)$, Algorithm \ref{alg:mwu} \textsc{MWU}($\BB, \pp, \mathcal{O}, \eta, \eps$) returns 
a feasible random vector $\zz \in \{\pm 1\}^n$ whose distribution $\mathcal{D}$ satisfies 
\begin{align}
f_{\BB}(\mathcal{D}) = 
\| \cov_{\zz \sim \mathcal{D}}(\BB \zz) \| \le (1+\eps) \eta.
\label{eqn:thm_mwu_f}    
\end{align}
In addition, the number of calls to $\mathcal{O}$ and the algorithm's runtime are polynomial in $m,n,\eps^{-1}$.
\label{thm:mwu}
\end{theorem}

\begin{algorithm}[htb]
\caption{\textsc{MWU}($\BB, \pp, \mathcal{O}, \eta, \eps$)}
\label{alg:mwu}
\begin{algorithmic}[1]
\STATE Set $\WW_0 \leftarrow \II \in \mathbb{R}^{m \times m}, \alpha \leftarrow 0$ and $t \leftarrow 1$.

\WHILE{$\alpha < \frac{2\ln m}{\eps \eta}$\label{lin:while}}  

\STATE Let $\zz_t \leftarrow \mathcal{O}(\BB, \WW_{t-1}, \pp)$, $\alpha_t \leftarrow \frac{\eps}{6 \| \cov(\BB \zz_t)\| }$.
\label{lin:alphat}

\STATE Update \vspace{-0.2cm}
\[
\WW_t \leftarrow \exp \left(\sum_{\tau=1}^t \alpha_{\tau} \cdot 
 \cov(\BB \zz_{\tau}) \right).
\]
\label{lin:wt}

\STATE Let $\alpha \leftarrow \alpha + \alpha_t$ and $t \leftarrow t+1$.

\ENDWHILE

\STATE Return a random $\zz$ sampled from $\{\zz_1, \ldots, \zz_{t-1} \}$ with $\Pr(\zz = \zz_{\tau}) = \frac{\alpha_{\tau} } {\alpha}$ for $\tau=1,\ldots,t-1$. \label{lin:sample}
\end{algorithmic}
\end{algorithm}

Algorithm \ref{alg:mwu} can incorporate additional constraints on random assignments. Given an oracle $\mathcal{O}$ that produces a random assignment $\zz'$ satisfying Equation \eqref{eqn:thm_mwu_oracle_condition} subject to the additional constraints, we can run Algorithm \ref{alg:mwu} with $\mathcal{O}$ to obtain a random assignment $\zz$ that satisfies Equation \eqref{eqn:thm_mwu_f} and the constraints.

\subsection{The MWU Oracle}
\label{sec:oracle}

We describe a polynomial-time algorithm for the oracle $\mathcal{O}$ that satisfies the condition in Theorem \ref{thm:mwu}. Our algorithm is presented in Algorithm \ref{alg:oracle} \textsc{Oracle}($\BB, \WW, \pp, \eps$).
Without loss of generality, we make the following assumptions on $\WW$: (1) the trace of $\WW$ is $1$; (2) $\WW$ is a diagonal matrix (otherwise, we take the eigendecomposition of $\WW$ and apply a linear transform to $\WW$ and $\BB$).

Algorithm \ref{alg:oracle} \textsc{Oracle} is inspired by algorithmic discrepancy theory, particularly the random walks over $[-1,1]^n$ from \cite{BDG16}.
It starts at $\zz_0 = 2\pp - \one$, the expected value of a feasible assignment. 
At each step $t$, it randomly moves from $\zz_{t-1}$ to a new position $\zz_t \in [-1,1]^n$, which ensures that $\mathbb{E}[\zz_t | \zz_{t-1}] = \zz_{t-1}$ (that is, $\zz_0, \zz_1, \ldots $ form a martingale).
Once the walk reaches a face of the cube $[-1,1]^n$, it remains on that face in all future steps. 
After enough steps, the walk reaches a corner of the cube at $\zz_T \in \{\pm 1\}^n$ and returns $\zz_T$.
We have $\mathbb{E}[\zz_T] = 2\pp - \one$.
The main challenge is how to properly choose $\zz_t$ given $\zz_{t-1}$ at each step $t$.

At each step $t$, we update \(\zz_t\) using the formula $\zz_t = \zz_{t-1} + \gamma_t \yy_t$ where $\yy_t \in \mathbb{R}^{n}$ is a unit vector and $\gamma_t \in \mathbb{R}$ is a random step size. 
Ideally, this update causes only a slight increase in the objective \(\langle \cov(\BB \zz), \WW \rangle\) while making significant progress in moving \(\zz\) toward a corner.
Below, we explain how to choose $\yy_t$ and $\gamma_t$ in detail, which differs from \cite{BDG16}.
We maintain a set $A_t$
of ``alive'' variables at the start of step $t$, as defined in line \ref{lin:live} of Algorithm \ref{alg:oracle}. 
Only alive variables can be changed.

Let $\tilde{\BB}_t$ be the submatrix of $\BB$ restricted to columns indexed in $A_t$.
We classify the rows of $\tilde{\BB}_t$ into ``big'' rows and ``light'' rows as follows:
Let\footnote{For a matrix $\AA$, let $\AA(S,T)$ denote the submatrix of $\AA$ restricted to rows in $S$ and columns in $T$; let $\AA(S,:)$ denote the submatrix restricted to rows in $S$, and $\AA(:,T)$ the submatrix restricted to columns in $T$.} 
$$B_t = \{j \in [m]: \|\tilde{\BB}_t(j,:) \|^2 > 1 + 1/\eps \}$$ 
be a set of rows with large norms, and let $L_t = [m] \setminus B_t$.
We choose $\yy_t$ (when restricted to alive entries) to be orthogonal to the big rows and have small projections (in absolute value) onto the light rows.

To formalize these concepts, we introduce the following notations (illustrated in Figure \ref{fig:algo_matrix}):
\begin{align}
\begin{split}
& \VV_{t,b} = \tilde{\BB}_t(B_t, :), ~ 
\VV_{t,l} = \tilde{\BB}_t(L_t, :), 
~ \WW_{t} = \WW(L_t, L_t), \\
& \MM_t = (1+\eps) \diag(\VV_{t,l}^\top \WW_t \VV_{t,l}) - \VV_{t,l}^\top \WW_t \VV_{t,l}.
\end{split}
\label{eqn:algo_mat}
\end{align}
We choose $\yy_t$ such that 
\begin{align}
\begin{split}
& \yy_t(A_t) \leftarrow \argmax_{\tilde{\yy} \in \mathbb{R}^{\abs{A_t}}}\{\tilde{\yy}^\top \MM_t \tilde{\yy}: \VV_{t,b} \tilde{\yy} = {\bf 0}, \norm{\tilde{\yy}} = 1 \} \\
& \yy_t(i) \leftarrow 0, ~ \forall i \notin A_t
\label{eqn:y_t}    
\end{split}
\end{align}
Next, we select the step size \(\gamma_t\) as a zero-mean random variable that pushes at least one of the alive variables to $\pm 1$ (thus not alive next step).
\begin{figure}[htb]
\centering
\includegraphics[width=0.4\textwidth]{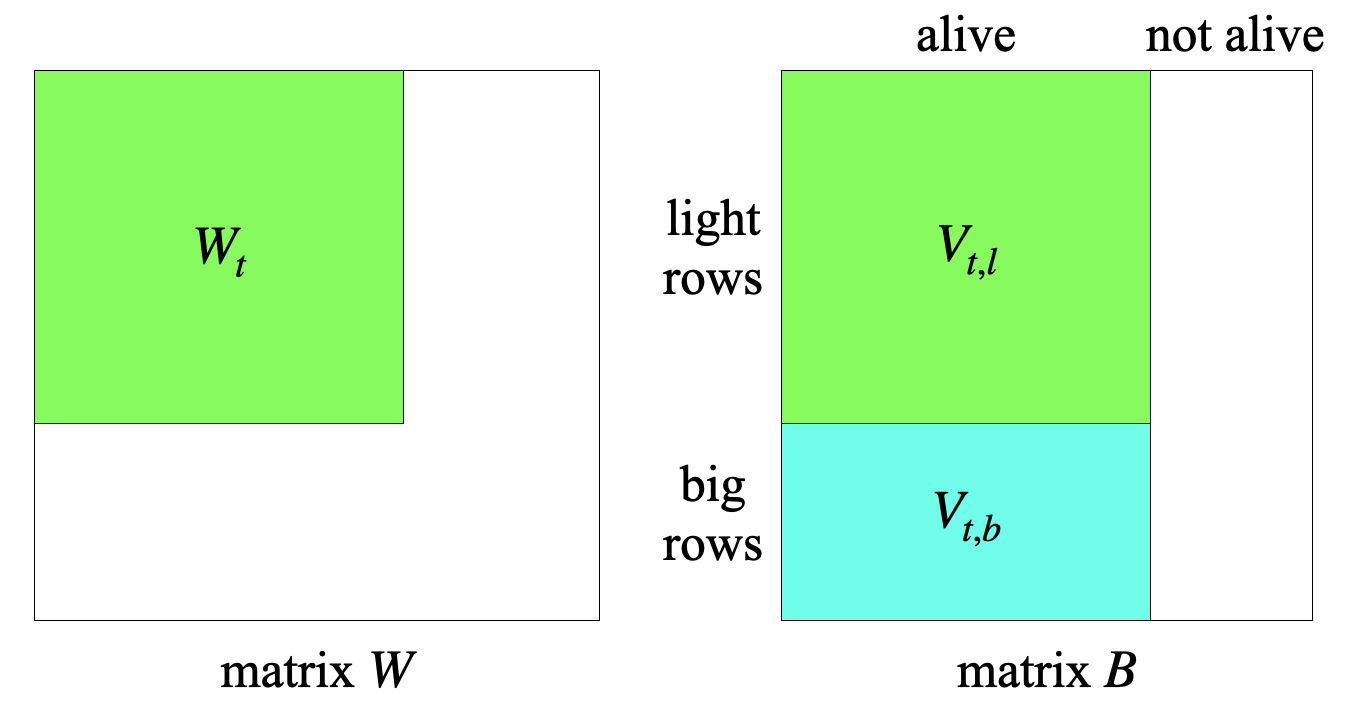}
\caption{A visualization of the matrices defined in Equation \eqref{eqn:algo_mat}. We reordered the columns and rows of $\BB$ and $\WW$ for better visualization.}
\label{fig:algo_matrix}
\end{figure}

\begin{theorem}
There exists an oracle, presented in Algorithm \ref{alg:oracle} \textsc{Oracle}($\BB, \WW, \pp, \eps$), that satisfies the conditions in Theorem \ref{thm:mwu} with 
\[
\eta \le (1+\eps) \min \left\{U_W, 1 + \frac{1}{\eps}
\right\},
\]
where $U_W = \langle \cov_{\yy \sim \mathcal{D}_{Bernoulli}}(\BB \yy), \WW \rangle$.
In addition, Algorithm \ref{alg:oracle}'s runtime is polynomial in $m,n,\eps^{-1}$.
\label{thm:oracle_algo}
\end{theorem}

Our proof for Theorem \ref{thm:oracle_algo} adopts a similar potential function to track the progress of our random walk, as in \cite{BDG16}.
However, our choice of the update vector $\yy_t$ in Equation \eqref{eqn:y_t} and our analysis details differ from those in \cite{BDG16}. 
\cite{BDG16} only shows that $\langle \cov(\BB \zz), \WW \rangle = O(1)$, a bound comparable to GSW. It is unclear how their algorithm compares with the Bernoulli design for unequal assignment probabilities.

\begin{algorithm}[htb]
    \caption{\textsc{Oracle}($\BB, \WW, \pp, \eps$)}
    \label{alg:oracle}

\begin{algorithmic}[1]

\STATE Set $\zz_0 \leftarrow 2\pp - \one$, $ t \leftarrow 1$, and $A_1 \leftarrow \{i \in [n]: \abs{\zz_{0}(i)} < 1 \}.$

\WHILE{$A_t$ is not empty} 

\STATE Find an update direction $\yy_t$ by Equation \eqref{eqn:y_t}.

\STATE Find a step size $\gamma_t$ by setting 
\begin{align*}
& \gamma_{+} \leftarrow \max\{\gamma: \zz_{t-1} + \gamma \yy_t \in [-1,1]^{n} \} \\
& \gamma_{-} \leftarrow \max\{\gamma: \zz_{t-1} - \gamma \yy_t \in [-1,1]^{n} \}
\end{align*}
and randomly sampling 
\[
\gamma_t = \left\{ \begin{array}{ll}
  \gamma_+   & \quad \text{with probability } \frac{\gamma_-}{\gamma_+ + \gamma_-} \\
 - \gamma_-   & \quad \text{with probability } \frac{\gamma_+}{\gamma_+ + \gamma_-}
\end{array}
\right.
\]

\STATE Update
$\zz_t \leftarrow \zz_{t-1} + \gamma_t \yy_t.$

\STATE Let $t \leftarrow t + 1$ and $A_t \leftarrow \{i \in [n]: \abs{\zz_{t-1}(i)}  < 1\}.$ \label{lin:live}

\ENDWHILE

\STATE Return $\zz_{t-1}$.

\end{algorithmic}
\end{algorithm}

We run Algorithm \ref{alg:mwu} \textsc{MWU}($\BB, \pp, \textsc{Oracle}, \eta, \eps$). \textsc{Oracle} is given in Algorithm \ref{alg:oracle} with the same error parameter $\eps$, and parameter $\eta$ is given in Theorem \ref{thm:oracle_algo}.
Combining Theorems \ref{thm:mwu} and \ref{thm:oracle_algo} results in Theorem \ref{thm:algo}.
The covariance matrix $\cov(\BB \zz_t)$ in Algorithm \ref{alg:mwu} might be unknown. In this case, we replace it with its empirical mean, and we discuss more details in Supplementary Section \ref{sec:appendix_mwu_known_cov}.
All proofs are in Supplementary Section \ref{sec:appendix_algo_proofs}.

%% file: experiment.tex
\section{EXPERIMENTS}
\label{sec:experiment}

In this section, we compare our design with several designs\footnote{Pairwise matching does not naturally generalize to unequal treatment-control assignments, as it pairs units and assigns one to treatment and the other to control.}: the GSW design, Bernoulli design, Complete Randomization, Randomized Block Design, and Rerandomization. 
We experiment with different treatment-control assignment probabilities by setting $\pp = p \one$, where $p$ ranges from $0.025$ to $0.975$. Due to the symmetry between the treatment and control groups, we only plot \( p \) from $0.5$ to $0.975$.
We also include experiments for $\pp$ with non-uniform entries in Supplementary Section \ref{sec:appendix_experiment_nonuniform_p}.
The designs are evaluated based on two metrics: (1) $\|\cov(\BB \zz) \|$, the objective of the DDM problem described in Problem \ref{prob}, and (2) the mean squared error (MSE) in estimating the ATE. 

Details of the design implementations are provided in Supplementary Section \ref{sec:appendix_design}. 
Our code for the MWU design is available at https://github.com/pengzhang91/MWU.

\subsection{The DDM Objective}
\label{sec:experiment_covariate}

We begin by examining different algorithms/designs for solving the DDM problem, where the goal is to minimize $\norm{\cov(\BB \zz)}$.
We consider two types of $\BB$: (1) randomly generated entries and (2) covariate data from the Lalonde dataset \citep{lalonde86}.

\textbf{Random $\BB$.} 
We consider two types of matrices for \(\BB\): (1) a matrix where all entries are i.i.d. random variables uniformly sampled from \([-1,1]\); (2) an augmented matrix as defined in Equation \eqref{eqn:mat_B}. For the first type, we set the dimensions of \(\BB\) to be \(20 \times 100\). For the second type, we set the dimensions of the covariate matrix \(\XX^\top\) in Equation \eqref{eqn:mat_B} to be \(40 \times 100\), with \(\XX\)'s entries being i.i.d. random variables uniformly sampled from \([-1,1]\); we set parameter $\phi=0.5$ for constructing $\BB$.
For each type of $\BB$, we generate independent samples of $\BB$ and plot the resulting $95\%$ confidence intervals in Figure \ref{fig:cov} (where the randomness comes only from the random samples of $\BB$). 
Among the six designs, the Bernoulli, Complete Randomization, and Randomized block designs perform the worst; Rerandomization is better; the MWU and GSW designs have the best results. We zoom in on the MWU and GSW designs, and we observe that MWU yields even better values of $\|\cov(\BB \zz)\|$ than GSW.

\begin{figure*}[htbp]
    \centering
    \begin{subfigure}{0.24\textwidth}
        \centering
        \includegraphics[width=\textwidth]{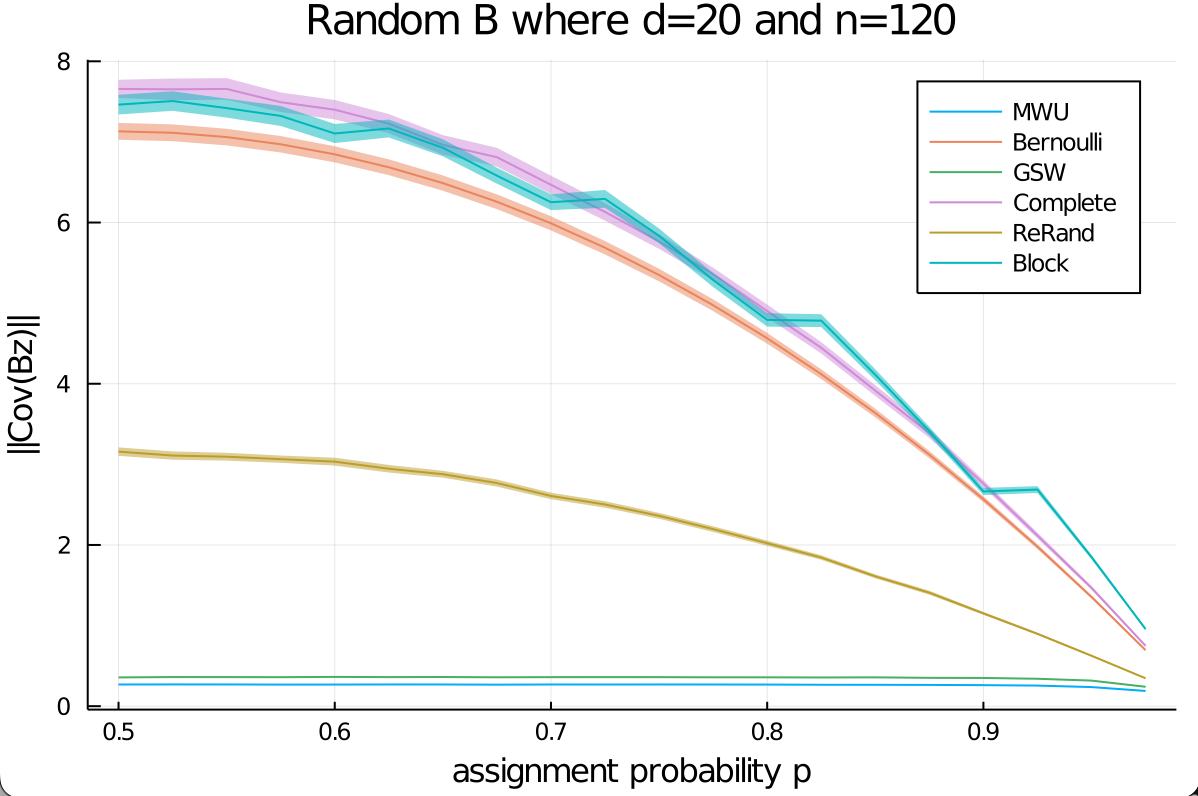}
    \end{subfigure}
    \hfill
    \begin{subfigure}{0.24\textwidth}
        \centering
        \includegraphics[width=\textwidth]{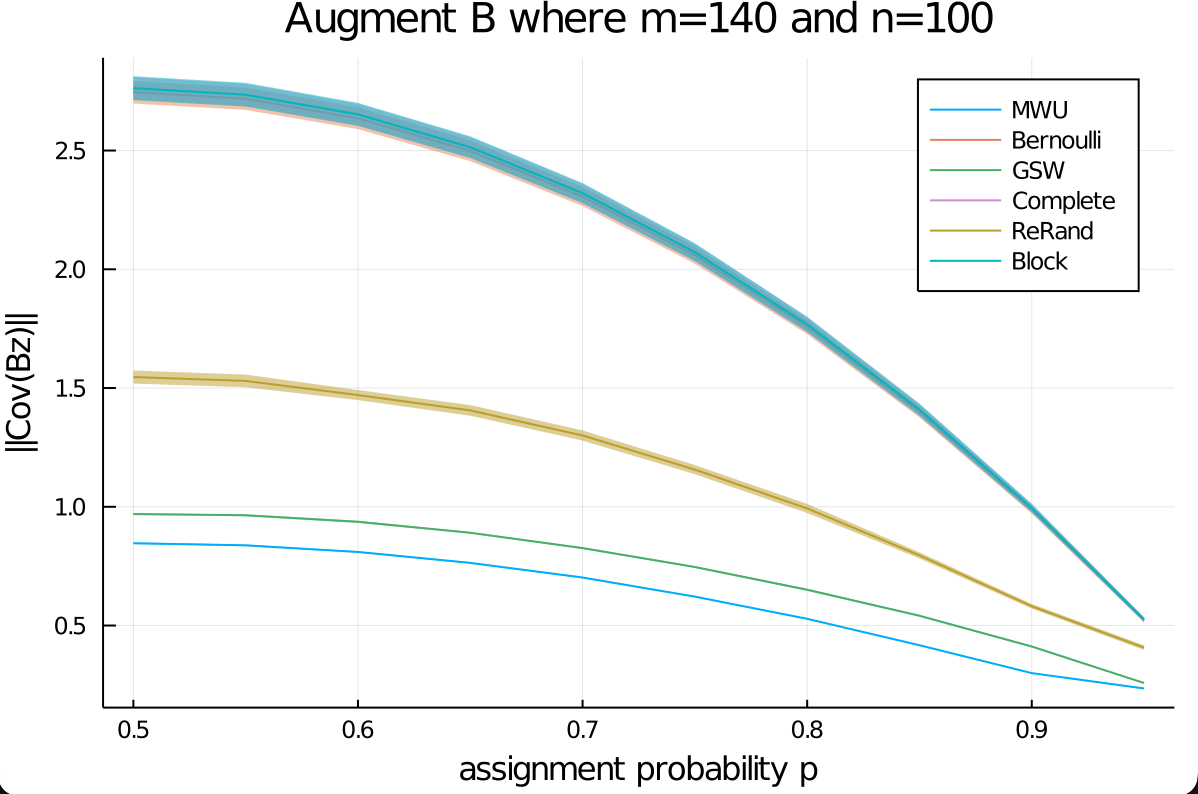}
    \end{subfigure}
    \hfill
    \begin{subfigure}{0.24\textwidth}
        \centering
        \includegraphics[width=\textwidth]{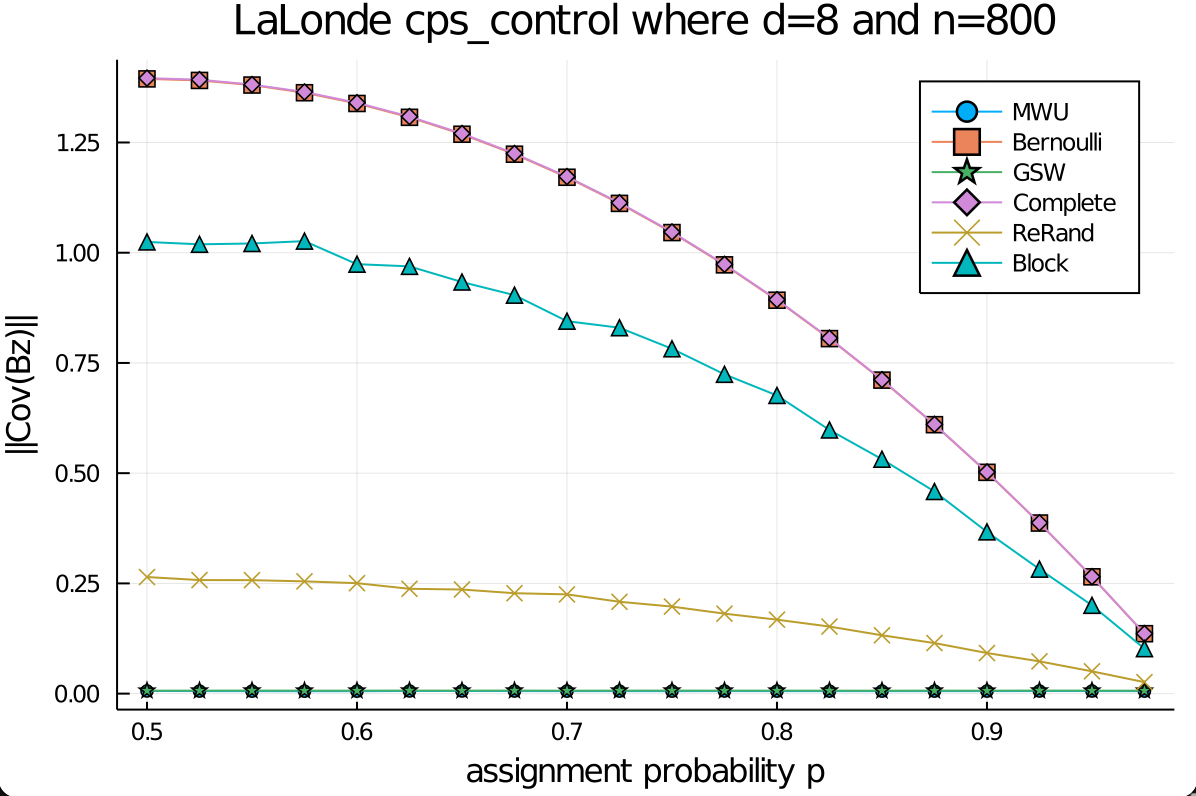}
    \end{subfigure}
    \hfill
    \begin{subfigure}{0.24\textwidth}
        \centering
        \includegraphics[width=\textwidth]{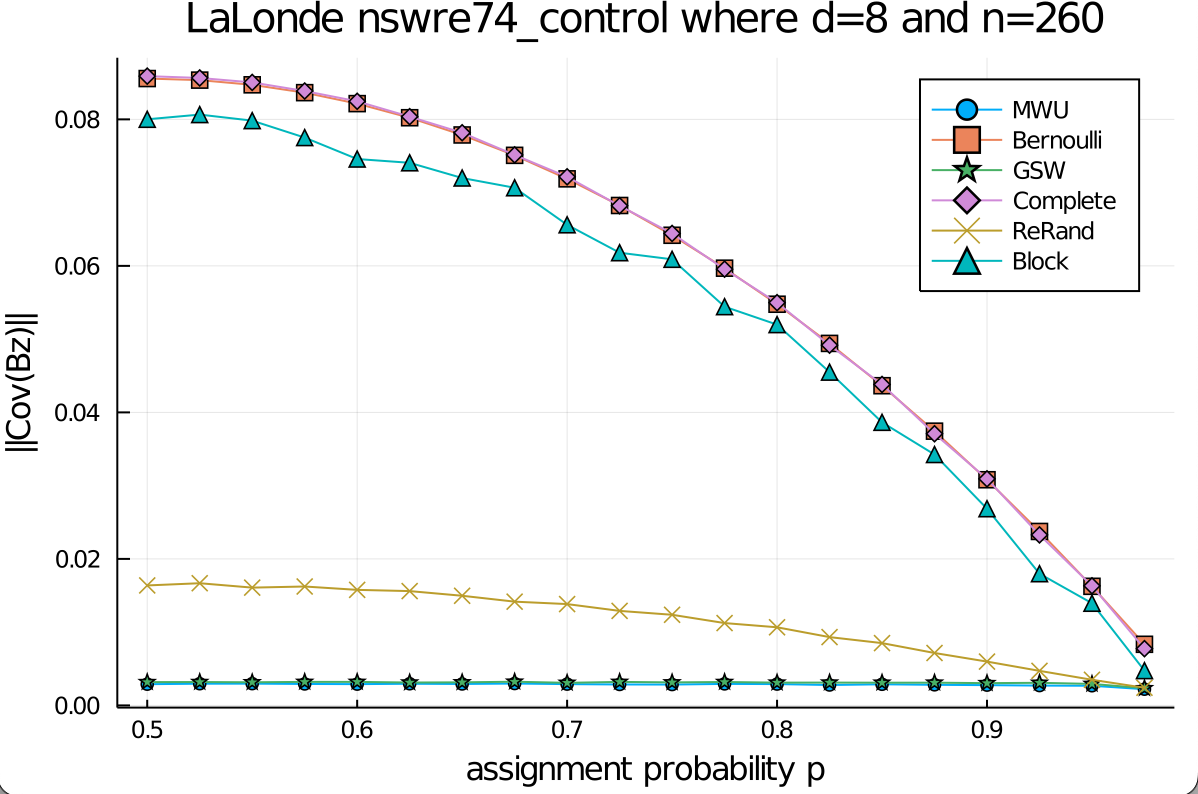}
    \end{subfigure} \\
    \begin{subfigure}{0.24\textwidth}
        \centering
        \includegraphics[width=\textwidth]{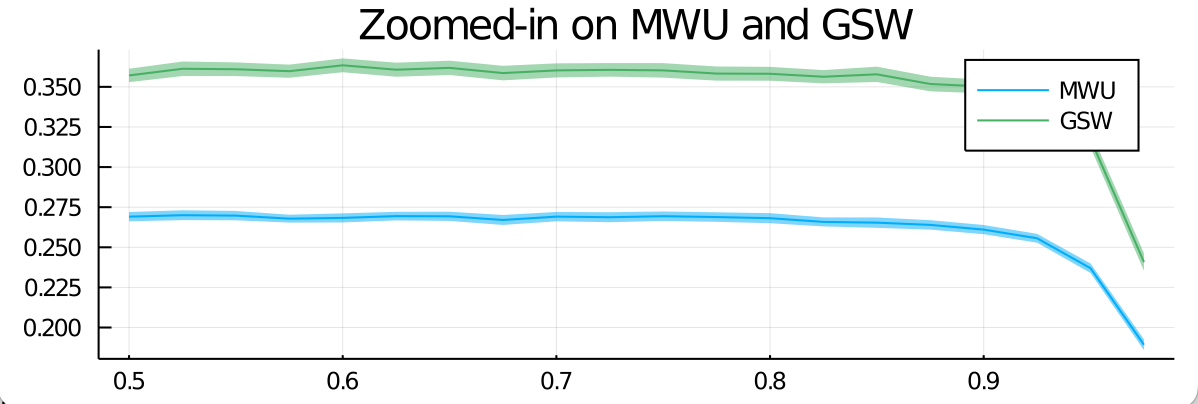}
    \end{subfigure}
    \hfill
    \begin{subfigure}{0.24\textwidth}
        \centering
        \includegraphics[width=\textwidth]{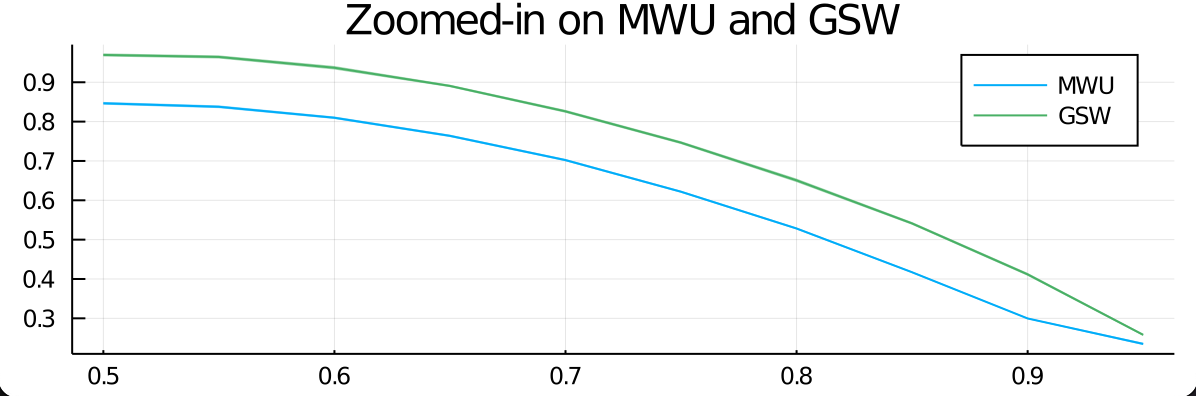}
    \end{subfigure}
    \hfill
    \begin{subfigure}{0.24\textwidth}
        \centering
        \includegraphics[width=\textwidth]{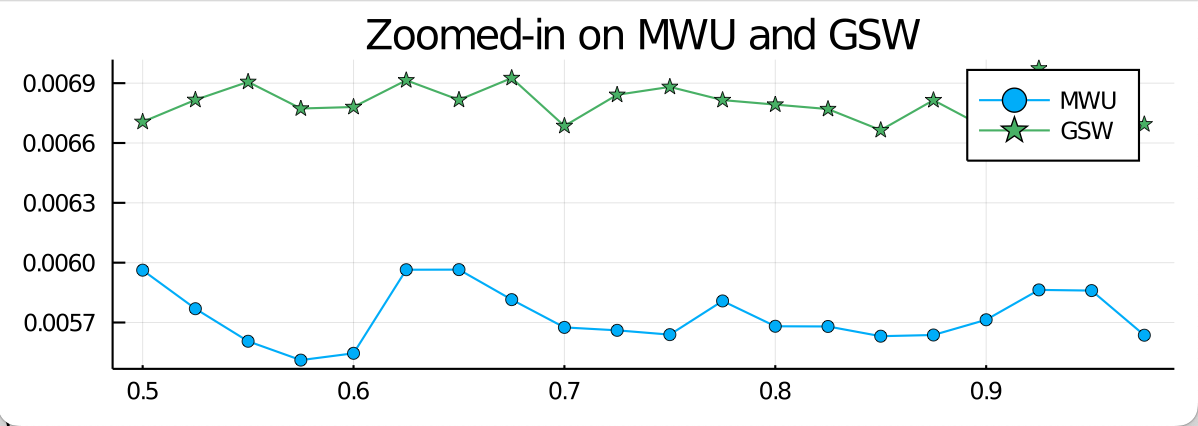}
    \end{subfigure}
    \hfill
    \begin{subfigure}{0.24\textwidth}
        \centering
        \includegraphics[width=\textwidth]{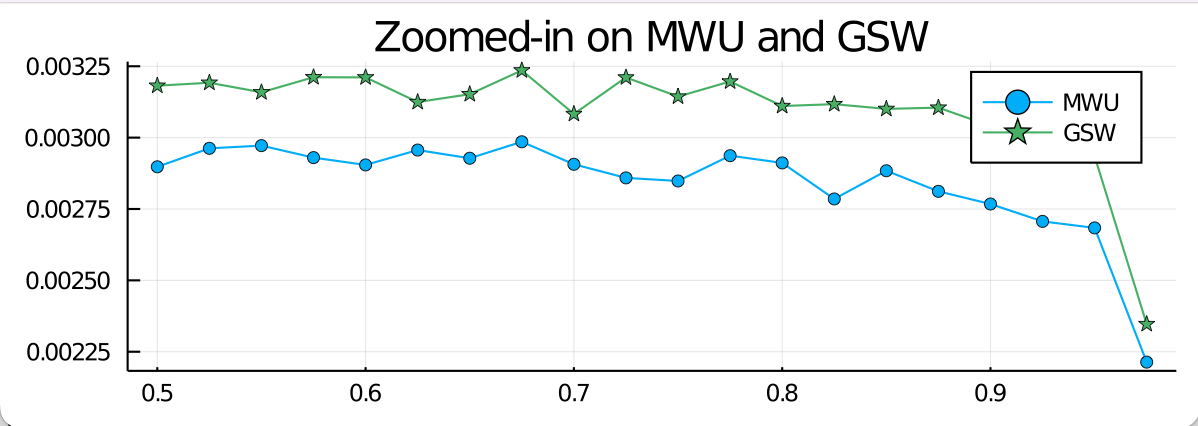}
    \end{subfigure}
    \caption{The DDM Objective.}
    \label{fig:cov}
\end{figure*}

\begin{figure*}[htbp]
    \centering
    \begin{subfigure}{0.24\textwidth}
        \centering
        \includegraphics[width=\textwidth]{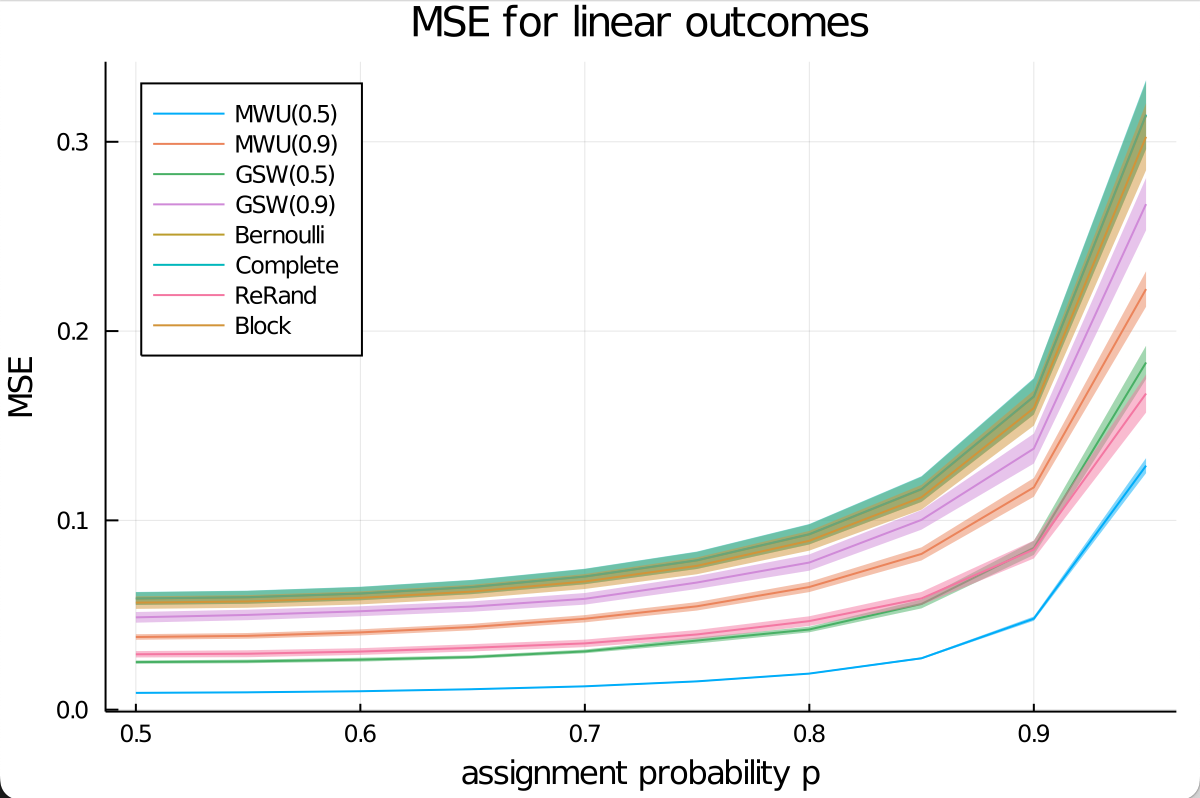}
    \end{subfigure}
    \hfill
    \begin{subfigure}{0.24\textwidth}
        \centering
        \includegraphics[width=\textwidth]{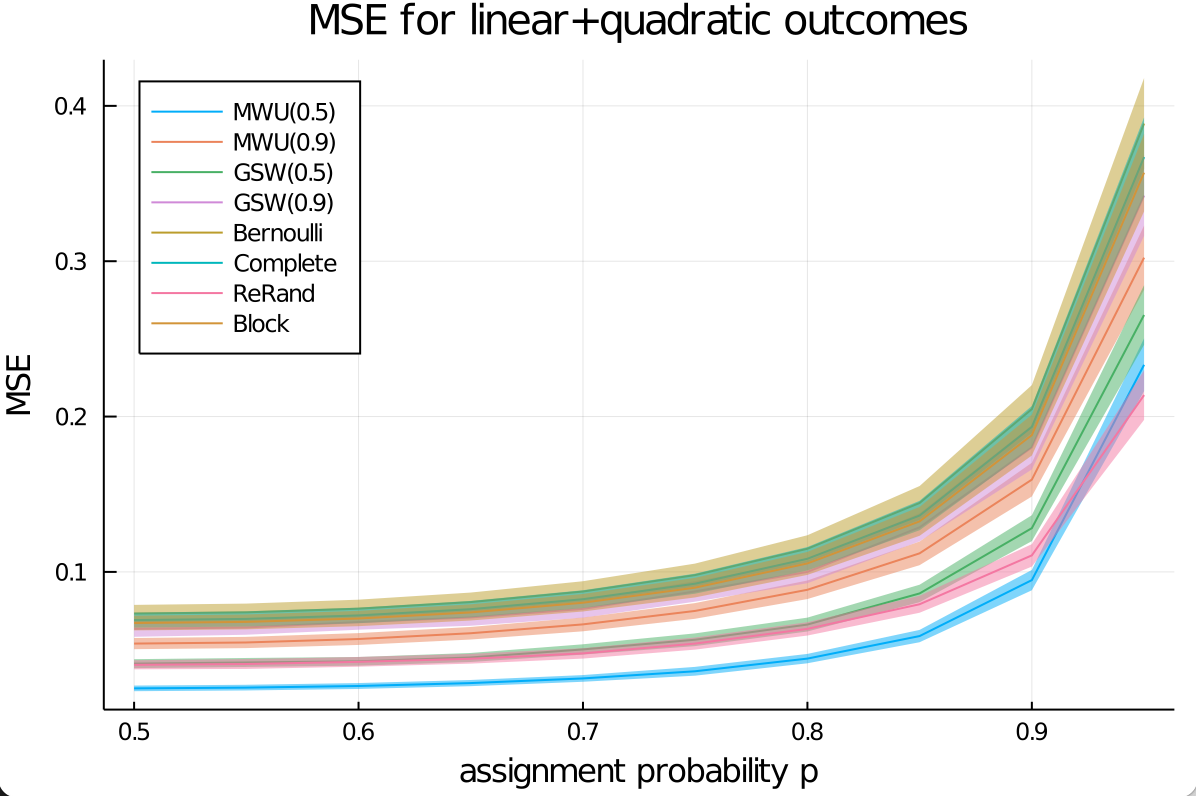}
    \end{subfigure}
    \hfill
    \begin{subfigure}{0.24\textwidth}
        \centering
        \includegraphics[width=\textwidth]{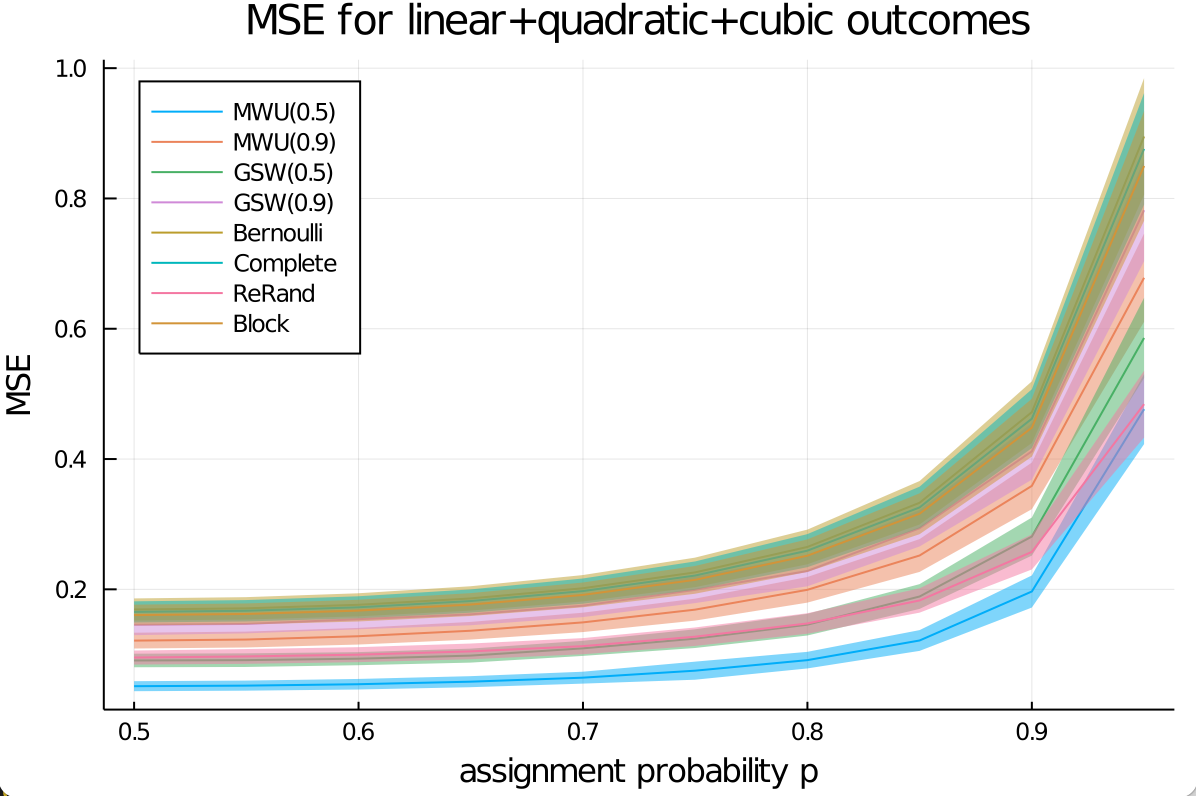}
    \end{subfigure}
    \hfill
    \begin{subfigure}{0.24\textwidth}
        \centering
        \includegraphics[width=\textwidth]{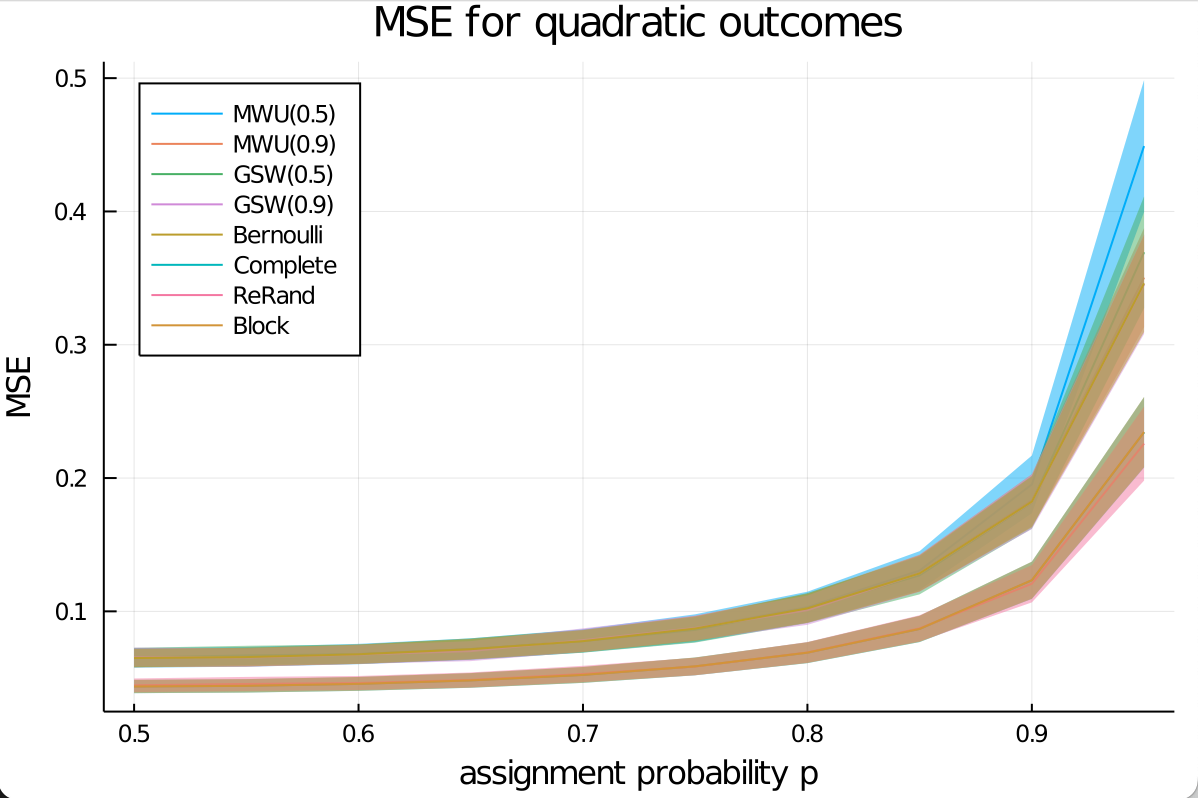}
    \end{subfigure} \\
    \begin{subfigure}{0.24\textwidth}
        \centering
        \includegraphics[width=\textwidth]{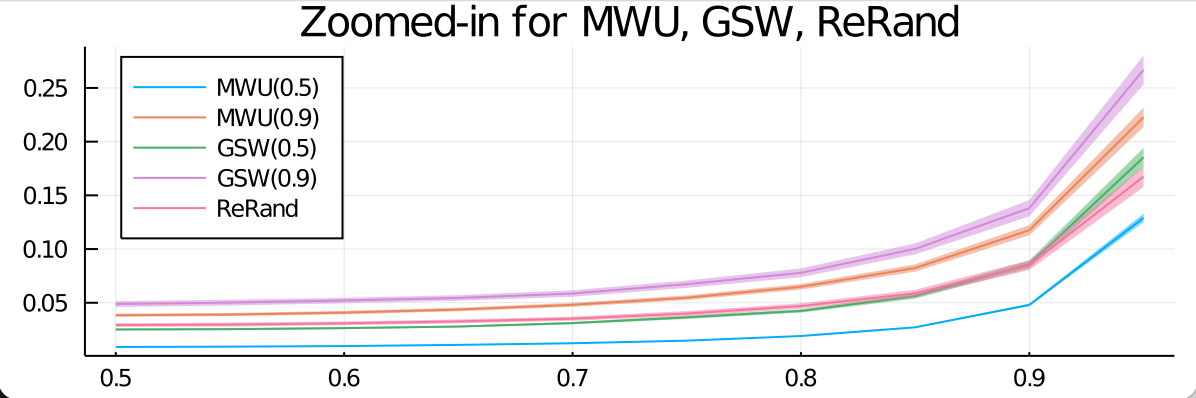}
    \end{subfigure}
    \hfill
    \begin{subfigure}{0.24\textwidth}
        \centering
        \includegraphics[width=\textwidth]{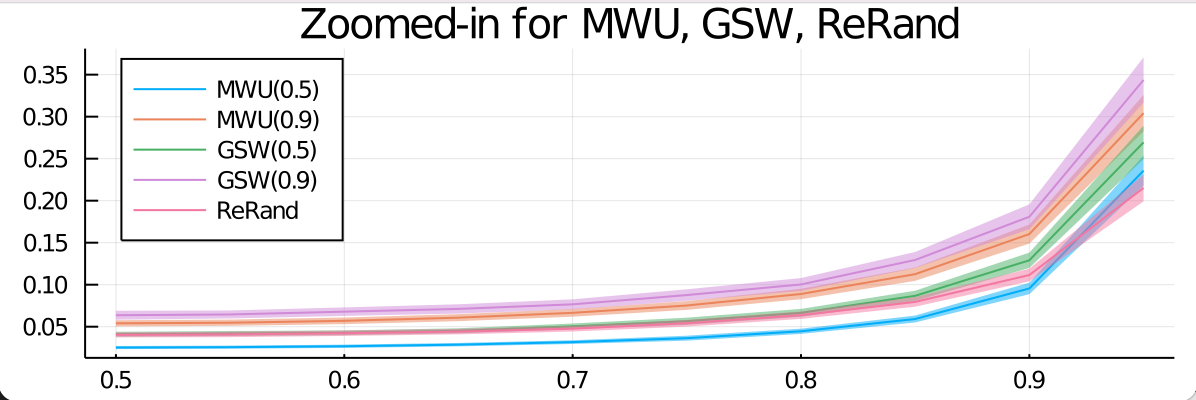}
    \end{subfigure}
    \hfill
    \begin{subfigure}{0.24\textwidth}
        \centering
        \includegraphics[width=\textwidth]{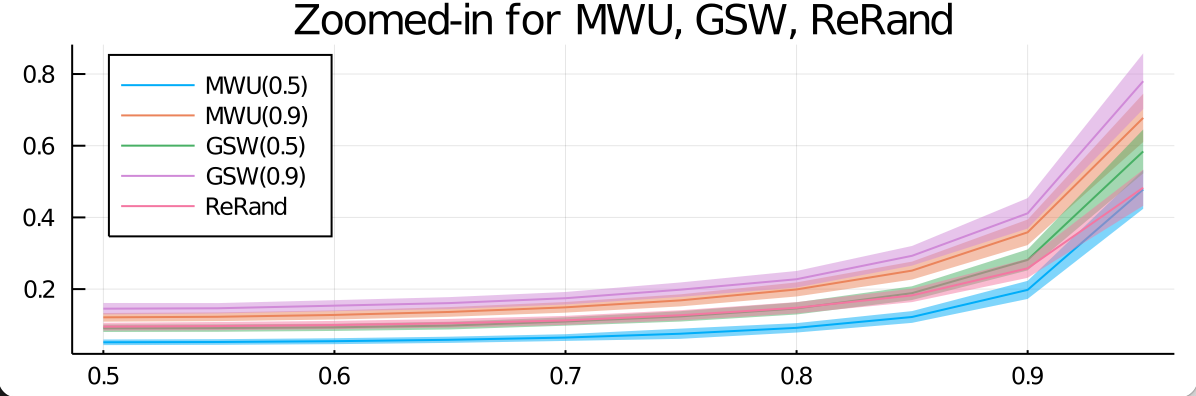}
    \end{subfigure}
    \hfill
    \begin{subfigure}{0.24\textwidth}
        \centering
        \includegraphics[width=\textwidth]{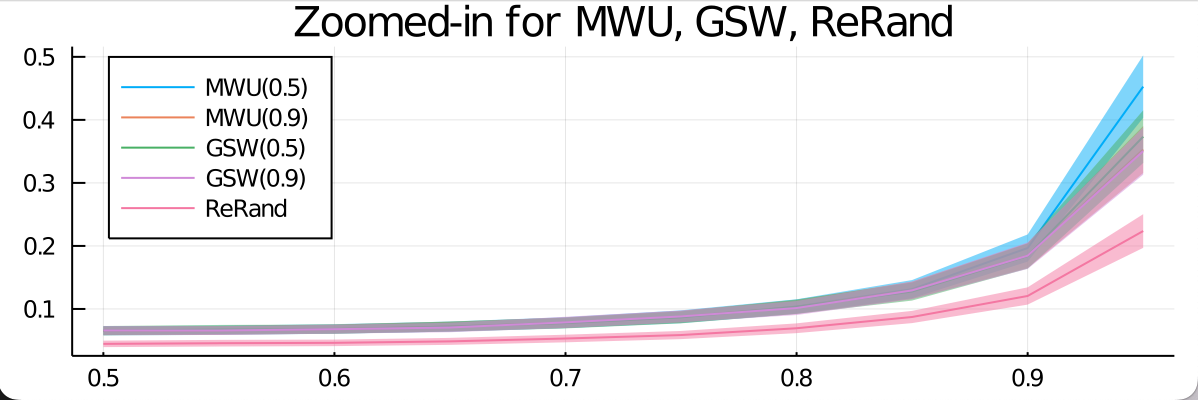}
    \end{subfigure}
    \caption{MSE (Mean-Squared Error) of Estimating ATE.}
    \label{fig:mse}
\end{figure*}

\textbf{LaLonde Dataset.} 
We evaluate the six designs using the LaLonde dataset from  \cite{lalonde86} and \cite{DW99,DW02} \footnote{The dataset is available at users.nber.org/\raisebox{0.5ex}{\texttildelow}rdehejia/data/.nswdata2.html. We use the data from files cps\_controls.txt and nswre74\_control.txt.
}. 
The dataset estimates the impact of the National Support Work Demonstration (NSW) job training program on trainee earnings. 
The dataset has different experimental and control data.
We test on two of them: CPS control data and NSW control data. Both the two datasets have eight covariates: four binary covariates and four numeric covariates.
For the CPS dataset, we randomly choose $800$ units (aka, $n=800$); for the NSW dataset, we use all the $260$ units.
We normalize each covariate to have a sample mean of $0$ and a sample standard deviation of $1$.
We take $\BB$ to be this standardized covariate matrix.
We then add an independent Gaussian noise $\mathcal{N}(0,0.02^2)$ to each covariate to make $\BB$ full row rank. Finally, we scale all the entries of $\BB$ to ensure the new matrix satisfies $\|\BB\|_{1,2} \le 1$.
Our experiment results are shown in Figure \ref{fig:cov}.
MWU has the lowest $\|\cov(\BB \zz) \|$ for all values of $p$.

While for $p = 1/2$ MWU and GSW perform similarly in the worst-case scenario, our experiment results show that MWU outperforms GSW in many practical cases.
In fact, Theorem \ref{thm:mwu} can be strengthened to show that \(\|\cov(\BB \zz)\| \leq (1+\epsilon) \eta'\) where
\[
\eta' = \max_{t=1,\ldots,T} \frac{\langle \text{Cov}(\BB \zz_t), \WW_{t-1} \rangle}{\tr(\WW_{t-1})},
\]
with \(\WW_{t-1}\) being the weight matrices generated during the MWU iterations (Algorithm \ref{alg:mwu}, Line \ref{lin:wt}).
In the worst case, $\eta' = \eta$ given in Equation \eqref{eqn:thm_mwu_oracle_condition} of Theorem \ref{thm:mwu}. However, in practice, $\eta'$ can be smaller than $\eta$, leading to improved performance.
In addition, the MWU oracle has a much simpler objective than GSW, which may contribute to better practical performance even for $p = 1/2$.

\subsection{Mean-Squared Error (MSE)}
\label{sec:experiment_mse}

We then evaluate all six designs by measuring the MSE for estimating the average treatment effect, as described in Section \ref{sec:problem_set}. 
We test the MWU and GSW designs using the design parameter $\phi \in \{0.5, 0.9\}$, following the recommendation from \cite{HSSZ19}, who suggests choosing $\phi \ge 0.5$ to ensure robustness.

We set $d=40$ and $n=100$, and generate the covariate $\xx_i$'s whose entries are i.i.d. random variables uniformly sampled from $[-1,1]$.
We choose potential outcomes $a_i = f(\xx_i)$ and $b_i = f(\xx_i) + \eps_i$, where $f(\xx_i)$ is a function that depends only on the first twenty covariates of unit $i$, and $\eps_i$ is Gaussian noise. 
We consider $f$ to take various forms, representing different relations between covariates and outcomes: linear, a mix of linear and quadratic, a mix of linear and quadratic and cubic, and pure quadratic terms. Details of the data-generating process are in Supplementary Section \ref{sec:appendix_data_generate}. The results are presented in Figure \ref{fig:mse}.

The MWU design with a parameter of $0.5$ (referred to as MWU($0.5$)) achieves the lowest MSE when the relationship between outcomes and covariates is linear (first plot in Figure \ref{fig:mse}) or nearly linear (second and third plots). Specifically, the best MSE from other designs is at least $2.61$ times that of MWU($0.5$) on average, and at least $3.00$ times for $p \in [0.25, 0.85]$. When the outcome-covariate relationship is pure quadratic (fourth plot), Rerandomization has the lowest MSE, but MWU($0.9$) performs comparably, with an MSE that is no more than $1.49$ times that of Rerandomization.

%% file: my_supplement.tex
\onecolumn

\aistatstitle{
Supplementary Materials}

\setcounter{section}{0}

\section{Instances where Bernoulli Performs Better than GSW}
\label{sec:numeric_instance}

In this section, we present two examples where the Bernoulli design performs better than the GSW design.

\subsection{A Numerical Example}

We present a numerical example showing that the Bernoulli design has a smaller objective value than the GSW design for the distributional discrepancy minimization (DDM) problem:
\begin{align*}
\BB = \begin{pmatrix}
1.0 & 0.0 & 0.0 \\
-0.375542 & -0.378668 & -0.845919 \\
0.558774 & -0.337221 & -0.757663    
\end{pmatrix}
\text{ and }
\pp = \begin{pmatrix}
0.925 \\ 0.925 \\ 0.945
\end{pmatrix}.
\end{align*}
The Bernoulli design has $\|\cov(\BB \zz_{Bernoulli}) \| = 0.38$, while the GSW design's is $\|\cov(\BB \zz_{GSW}) \| = 0.41$. 
We can generalize this $3$-by-$3$ matrix to a block diagonal matrix $\BB$ in larger dimensions. 
In general, the Bernoulli design may perform better than the GSW when entries of $\pp$ are near $1$ or $0$ and the operator norm of $\BB$ is small but greater than $1$.

\subsection{An Artificial Example}

We then provide an example where the Bernoulli and GSW algorithms can differ arbitrarily.
A description of the GSW algorithm can be found in \cite{BDGL18} or \cite{HSSZ19}.

Let $n$ be a positive integer such that $m = \frac{n}{3}$ is a power of $2$. Define a matrix $\BB \in \mathbb{R}^{m \times n}$ so that in the $i$-th row, the entries at positions $(3i - 2)$ and $(3i - 1)$ are $1/2$, the entry at position $(3i)$ is $-1$, and all other entries are $0$. Define the assignment probabilities $\pp = (1-\delta) \one \in \mathbb{R}^n$ where $\delta = \frac{3}{2n}$.
Then, $f_{\BB}(\mathcal{D}_{Bernoulli}) = 4(1 - \delta)\delta \| \BB \BB^\top\| = \Theta(1/n)$.

Consider running the GSW algorithm. GSW conducts a random walk within $[-1,1]^n$, starting at the point $\zz_0 = 1 - 2\delta$. In its first iteration, GSW selects an update vector $\yy_1$ whose first half of the entries are $1$ and whose second half are $-1$. After this iteration, we obtain $\zz_1$. With probability $1/2$, the first half of $\zz_1$ is $1$ and the second half is $1 - 4\delta$, and with probability $1/2$, the first half of $\zz_1$ is $1 - 4\delta$ and the second half is $1$. In either case, half of the entries of $\zz_1$ are $1 - 4\delta$ and the rest are $1$, and $\BB(\zz_1 - \zz_0) = {\bf 0}$.
Continuing in this manner, after the second iteration, a quarter of the entries of $\zz_2$ will be $1 - 8\delta$ and the rest will still be $1$. After $k = \log_2(n/3)$ iterations, we reach a point $\zz_k$ where three entries are $1 - 2^{k+1}\delta = 0$ at positions $3i-2, 3i-1, 3i$ for some $i \in [m]$, and all remaining entries are $1$; in addition, $\BB(\zz_k-\zz_0) = {\bf 0}$. 

At this stage, the problem reduces to running GSW on the input matrix $(1/2,1/2,-1)$ with assignment probabilities of $0.5$ for each index (aka, $\zz_0 = (0,0,0)$). If the first iteration of GSW picks the update vector $\yy_1=(1,-1,0)$, then we get $\zz_1 = (1,-1,0)$ or $(-1,1,0)$.
The second iteration will lead to $f_{\BB}(\mathcal{D}_{GSW}) = \Theta(1)$. It is significantly larger than $f(\mathcal{D}_{Bernoulli}) = \Theta(1/n)$.

In this example, there may be multiple choices for the update vector $\yy_t$ at each iteration of the GSW algorithm. If GSW happens to choose these update vectors poorly, it can produce an assignment $\zz$ with large $\|\cov(\BB \zz)\|$, leading to a large gap between the outcomes of the Bernoulli and GSW algorithms.

\input{appendix_lower}

\input{appendix_algo}

\section{Additional Experiment Details}
\label{sec:appendix_experiment}

We include the details of design implementations and experiment setups for the experiments in Section \ref{sec:experiment}.

\subsection{Designs}
\label{sec:appendix_design}

\subsubsection{The MWU Design}

We provide implementation details for our MWU design. Algorithm \ref{alg:mwu} has two parameters: $\eps$ and $\eta$. The parameter $\eps$ controls the approximation quality of the assignment returned by the oracle in Algorithm \ref{alg:oracle} and the assignment by Algorithm \ref{alg:mwu} itself. 
The parameter $\eta$ serves as an upper bound for the assignments generated by the oracle and, along with $\eps$, determines the number of iterations in the while loop at line \ref{lin:while} of Algorithm \ref{alg:mwu}.
Given that the theoretical upper bound in Theorem \ref{thm:oracle_algo} for Algorithm \ref{alg:oracle} might be overly conservative, we simplify the parameter choices by omitting $\eta$ and introducing a new parameter, $T$, to explicitly specify the number of while-iterations. 
This modification does not change the algorithm's main ideas. If a suitable $\eta$ is available, one can estimate an upper bound for the while-iterations and assign this value to $T$.

As explained in Supplementary Section \ref{sec:appendix_unknown_covariance}, the matrices $\cov(\BB \zz_t)$ in lines \ref{lin:alphat} and \ref{lin:wt} of Algorithm \ref{alg:mwu} typically do not have explicit forms. However, we can approximate these matrices by calculating their empirical mean using $N$ independent samples of $\zz_t$. 
Therefore, we include $N$ as an additional parameter in the implementation of Algorithm \ref{alg:mwu}.

In our experiments, we set $\eps = 0.2$, $T = 200$, and $N=50$.

\subsubsection{Other Benchmark Designs}

We provide the details for some compared benchmark designs.
Our implemented randomized block design follows the one in \cite{azriel2022optimality}: we only use the first two covariates to block units. Suppose we set each block to have size $n_B$. We first sort the units by the first covariate. Then within blocks of size $2n_B$, we sort and block the units by the second covariate.
We implement Rerandomization with $0.1\%$ (exact) acceptance probability and Mahalanobis distance.

\subsection{Data Generating Process}
\label{sec:appendix_data_generate}

We detail the outcome data-generating process in Section \ref{sec:experiment_mse}.
Recall that, for $n = 100$ units, we generate each covariate vector $\xx_i \in \mathbb{R}^{40}$ with i.i.d. entries uniform from $[-1,1]$.
We choose outcomes $a_i = f(\xx_i)$ and $b_i = f(\xx_i) + \eps_i$, where $\eps_i \sim \mathcal{N}(0,0.1^2)$. 
For each $i$, let 
\[
s(\xx_i) = \sum_{j=1}^{20} \xx_i(j),
\]
the sum of the first twenty entries of $\xx_i$.
We choose function $f$ from the following four categories:
\begin{enumerate}
\item linear: $f(\xx_i) = s(\xx_i)$
\item quadratic: $f(\xx_i) = s(\xx_i)^2$
\item a mix of linear and quadratic: $f(\xx_i) = s(\xx_i) + 0.5s(\xx_i)^2$
\item a mix of linear, quadratic and cubic: $f(\xx_i) = s(\xx_i) + 0.5s(\xx_i)^2 + 0.5s(\xx_i)^3$
\end{enumerate}

\subsection{Non-uniform Assignment Probabilities in $\pp$}
\label{sec:appendix_experiment_nonuniform_p}

We demonstrate similar experiment results when the portability vector $\pp$ entries are non-uniform. 
Specifically, we consider that half of the units have assignment probability $p_1$ and the other half have $p_2$, where $p_1 \neq p_2$.
We compare the MWU, GSW, and Bernoulli designs.
Complete randomization and Rerandomization do not naturally generalize to non-uniform $\pp$.

For the DDM problem, following our setup for the second figure in Figure \ref{fig:cov}, we let $\BB$ be the augmented matrix with $\phi=0.5$. 
First, we consider $p_1=0.1, p_2=0.5$. In this case, 
the DDM objective $\| \text{Cov}(\BB \zz)\|$ for MWU, GSW, and Bernoulli are $0.814, 0.904, 3.543$. 
Second, we consider $p_1=0.1, p_2=0.9$. In this case, 
the DDM objective $\| \text{Cov}(\BB \zz)\|$ for MWU, GSW, and Bernoulli are $0.300, 0.396, 1.623$.
MWU has the best performance.

For MSE (mean-squared error), we consider linear outcomes, linear+quadratic outcomes, linear+quadratic+cubic outcomes, and quadratic outcomes, as in Figure \ref{fig:mse}.
For probabilities $p_1=0.1, p_2=0.5$, 
the MSE of MWU(0.5), MWU(0.9), GSW(0.5), GSW(0.9), Bernoulli are
\begin{align*}
\text{linear: } & 0.041, 0.077, 0.061, 0.089, 0.111     \\
\text{lin+quad: } & 0.060, 0.108, 0.082, 0.118, 0.142 \\
\text{lin+quad+cube: } & 0.106, 0.212, 0.158, 0.234, 0.289 \\
\text{quadratic: } & 0.098, 0.108, 0.097, 0.103, 0.105
\end{align*}
For probabilities $p_1=0.1, p_2=0.9$, the MSE of MWU(0.5), MWU(0.9), GSW(0.5), GSW(0.9), Bernoulli 
are 
\begin{align*}
\text{linear: } 0.043, 0.112, 0.075, 0.134, 0.164 \\
\text{lin+quad: } 0.094, 0.182, 0.132, 0.209, 0.237 \\
\text{lin+quad+cube: } 0.180, 0.370, 0.267, 0.427, 0.490 \\
\text{quadratic: } 0.167, 0.176, 0.169, 0.172, 0.174
\end{align*}
MWU(0.5) performs best when outcomes are linear or almost linear in covariates. MWU(0.9) performs comparably when outcomes are quadratic in covariates.

%% file: appendix_lower.tex
\section{Missing NP-hardness Proofs for the DDM Problem}
\label{sec:appendix_hard}

In this section, we prove two strong NP-hardness results for the Distributional Discrepancy Minimization (DDM) problem described in Problem \ref{prob}. 
One result is for equal assignment probabilities (Theorem \ref{thm:main1}), and the other is for unequal probabilities (Theorem \ref{thm:main2_appendix}).
The results and proofs are originally presented in an unpublished manuscript by the second author \citep{zhang2022hardness}.

Let $C(\BB, \pp)$ be the optimal value of $f_{\BB}(\mathcal{D})$ defined in Equation \eqref{eqn:goal} of Problem \ref{prob}.

\begin{theorem}
There exists a universal constant $C_1 > 0$ such that given a matrix $\BB \in \mathbb{R}^{m \times n}$ with $\norm{\BB}_{1,2} \le 1$, it is NP-hard to distinguish whether $C(\BB, (1/2){\bf 1}) = 0$
or $C(\BB, (1/2){\bf 1}) > C_1$.
\label{thm:main1}
\end{theorem}

\begin{theorem}[Restatement of Theorem \ref{thm:main2}]
There exists a universal constant $C_2 > 0$ such that the following holds.
For any positive integer $n$ and parameters $\alpha \in (0,1/2), \beta \in (0,1)$, there exists 
$\pp \in \{1-\alpha, 1-2\alpha(1-\beta), 1-2\alpha \beta \}^n$ such that it is NP-hard to distinguish between the following two cases for a given matrix $\BB \in \mathbb{R}^{m \times n}$ with $\norm{\BB}_{1,2} \le 1$: (1) $C(\BB, \pp) = 0$ or (2) $C(\BB, \pp) > C_2 \alpha^2 (2\beta-1)^2$.
\label{thm:main2_appendix}
\end{theorem}

Our proofs are based on reductions from the $2$-$2$ \ssplit \ problem, which was introduced and shown to be NP-hard in a strong sense in \citet{guruswami04}.
In an instance of the \emph{$2$-$2$ \ssplit} problem, we are given a universe $U = \{1,2,\ldots,n\}$ and a family of sets 
$\calS = \{S_1, \ldots, S_m\}$ in which each $S_j$ consists of $4$ distinct elements from $U$. We denote such an instance $\mathcal{I}(U,\mathcal{S})$.
Our goal is to find an assignment of the $n$ elements in $U$ to $\{\pm 1\}$, denoted by $\yy \in \{\pm 1\}^n$, to maximize the number of sets in $\calS$ in which the values of its elements sum up to $0$.
We say an assignment $\yy$ \emph{$2$-$2$-splits} (or simply, splits) a set $S_j \in \calS$ if $\sum_{i \in S_j} \yy(i) = 0$; we say $\yy$ \emph{unsplits} $S_j$ otherwise, in which case $\sum_{i \in S_j} \yy(i) \in \{\pm 2, \pm 4\}$.
We say a $2$-$2$ \ssplit \  instance is \emph{satisfiable} if there exists an assignment that splits all the sets in $\mathcal{S}$.
For any $0 < \gamma < 1$, we say an instance is \emph{$\gamma$-unsatisfiable} if any assignment must unsplit at least $\gamma$ fraction of the sets in $\calS$.
A $2$-$2$ \ssplit \ instance is called a \emph{$(3,2$-$2)$ \ssplit} instance if each element in $U$ appears in at most $3$ sets in $\calS$. In such an instance, we have $m \le (3/4)n$.

\begin{theorem}[\citet{SZ22}]
There exists a constant $\gamma > 0$ such that it is NP-hard to distinguish satisfiable instances of the $(3,2$-$2)$ \ssplit \  problem from $\gamma$-unsatisfiable instances.
\label{thm:322sspit_hard}
\end{theorem}

Our proofs are inspired by the methods from \citet{CNN11} and \citet{SZ22}. However, the problems and proofs in these two papers are very different from ours.

\subsection{Proof of Theorem \ref{thm:main1}}
\label{sec:thm1}

We locally abuse our notations to let the columns of $\BB$ be $\xx_1, \ldots, \xx_N$ in this subsection.

Given a $(3,2$-$2)$ \ssplit \  instance $\mathcal{I}(U,\mathcal{S})$ where $\abs{U} = n$ and $\abs{\calS} = m$,
we will construct a matrix $\BB \in \mathbb{R}^{d \times N}$ where each column $\xx_i$ has Euclidean norm $1$ and $N,d$ are parameters to be determined later. 
Our construction will map a satisfiable $(3,2$-$2)$ \ssplit \  instance to an $\BB$ such that $C(\BB, (1/2) \one) = 0$ and a $\gamma$-unsatisfiable instance to an $\BB$ such that $C(\BB, (1/2) \one) > C_1$.

For each element $i \in U$, let $A_i \subset \{1,\ldots,m\}$ consist of the indices of the sets that contain $i$. 
For each element $i$ that appears in exactly $1$ set in $\calS$ (that is, $\abs{A_i} = 1$), we create $4$ new sets and $2$ new elements.
For each element $i$ that appears in $2$ sets in $\calS$, we create $5$ new sets and 
$3$ new elements. Let $B_i$ be the set consisting of the indices of the newly created sets for element $i$. The sets $B_i$'s are disjoint.
Suppose there are $n_1$ elements in $U$ that appear in exactly $1$ set in $\calS$
and $n_2$ elements that appear in $2$ sets.
We set 
\[
\dims = m + 4n_1 + 5n_2 \le m + 5n, ~ N = n + 2n_1 + 3n_2 \le 4n.
\]
Consider each element $i \in U$,
there are $3$ cases depending on how many sets in $\calS$ containing $i$:
\begin{enumerate}
  \item Element $i$ appears in $3$ sets in $\calS$: We define $\xx_i \in \mathbb{R}^d$ 
  such that $\xx_i(j) = \frac{1}{\sqrt{3}}$ for $j \in A_i$ and $\xx_i(j) = 0$ otherwise.

  \item Element $i$ appears in $1$ set in $\calS$: Suppose $B_i = \{i_1, i_2, i_3, i_4\}$.
  We define $\xx_i \in \mathbb{R}^d$ such that $\xx_i(j) = \frac{1}{\sqrt{3}}$ for $j \in A_i \cup \{i_1, i_2\}$
  and $\xx_i(j) = 0$ otherwise. We define two more vectors: (1) $\uu_{i,1} \in \mathbb{R}^d$ such that 
  $\uu_{i,1}(i_1) = \uu_{i,1}(i_3) = \uu_{i,1}(i_4) = \frac{1}{\sqrt{3}}$ and 
  $\uu_{i,1}(j) = 0$ for all other $j$'s, and (2) $\uu_{i,2} \in \mathbb{R}^d$ such that 
  $\uu_{i,2}(i_2) = -\frac{1}{\sqrt{3}}$ and $\uu_{i,2}(i_3) = \uu_{i,2}(i_4) = \frac{1}{\sqrt{3}}$ and
  $\uu_{i,1}(j) = 0$ for all other $j$'s.

  \item Element $i$ appears in $2$ sets in $\calS$: Suppose $B_i = \{i_1, i_2, i_3, i_4, i_5\}$.
  We define $\xx_i \in \mathbb{R}^d$ such that $\xx_i(j) = \frac{1}{\sqrt{3}}$ for $j \in A_i \cup \{i_1\}$
  and $\xx_i(j) = 0$ otherwise. We define three more vectors (1) 
  $\uu_{i,1} \in \mathbb{R}^d$ such that $\uu_{i,1}(i_1) = \uu_{i,1}(i_2) = \uu_{i,1}(i_3) = \frac{1}{\sqrt{3}}$ and
  $\uu_{i,1}(j) = 0$ for all other $j$'s, (2) $\uu_{i,2} \in \mathbb{R}^d$ such that 
  $\uu_{i,2}(i_2) = \uu_{i,2}(i_4) = \uu_{i,2}(i_5) = \frac{1}{\sqrt{3}}$ and 
  $\uu_{i,2}(j) = 0$ for all other $j$'s, and (3) $\uu_{i,3} \in \mathbb{R}^d$ such that 
  $\uu_{i,3}(i_3) = -\frac{1}{\sqrt{3}},  \uu_{i,3}(i_4) = \uu_{i,3}(i_5) = \frac{1}{\sqrt{3}}$ and 
  $\uu_{i,3}(j) = 0$ for all other $j$'s.
\end{enumerate}
We let $\xx_{n+1}, \ldots, \xx_{N}$ be the vectors $\uu_{i,h}$'s constructed above.
We can check that all $\xx_1, \ldots, \xx_N$ have Euclidean norm $1$. 

By our construction, the first $m$ entries of every vector $\uu_{i,h}$ all have a zero value.
For any assignment $\yy \in \{\pm 1\}^n$ for the $(3,2$-$2)$ \ssplit \  instance and $j \in \{1,\ldots, m\}$, the number $\sum_{i=1}^n \yy(i) \xx_i(j)$ equals the sum of the elements in set $S_j$.

\begin{claim}
For any $\yy \in \{\pm 1\}^n$, there exists a vector $\yy' \in \{\pm 1\}^N$ such that the following holds: Let $\ff = \sum_{i=1}^N \yy'(i) \xx_i$.
Then, $\ff(j) = \sum_{i=1}^n \yy(i) \xx_i(j)$ for every $j \in \{1,\ldots,m\}$, and $\ff(j) = 0$ for every $j \in \{m+1,\ldots,d\}$.
\label{clm:sign_u}  
\end{claim}
\begin{proof}
For each $i \in \{1,\ldots,n\}$, we set $\yy'(i) = \yy(i)$.
Since the first $m$ entries of $\xx_i$ for $i > n$ are all zero, our $\yy'$ satisfies the first condition in the statement for $\ff(j)$ where $j \in [m]$.
We will choose the signs of the rest of the entries of $\yy'$
to satisfy the second condition.

Since all the $B_i$'s are disjoint,
for each element $i$ appearing in less than $3$ sets in $\calS$, we only need to check the entries of $\ff$ with indices in $B_i$. 
Let $i \in U$ be an element that appears in $1$ set in $\calS$. 
The subvectors of $\xx_i, \uu_{i,1}, \uu_{i,2}$ restricted to the coordinates in $B_i$ are: 
\[
\begin{pmatrix}
1 \\
1 \\
0 \\
0
\end{pmatrix},
\begin{pmatrix}
1 \\
0 \\
1 \\
1
\end{pmatrix},
\text{ and }  
\begin{pmatrix}
0 \\
-1 \\
1 \\
1
\end{pmatrix}.
\]
We choose the signs in $\yy'$ for $\uu_{i,1}, \uu_{i,2}$ to be $-\yy(i)$ and $\yy(i)$, respectively, which guarantees the signed sum of the $\xx_1, \uu_{i,1}, \uu_{i,2}$ is ${\bf 0}$ when restricted to $B_i$. Since any other vector has $0$ for the coordinates in $B_i$, we have $\sum_{i=1}^N \yy'(i) \xx_i(j) = 0$ for $j \in B_i$. 
Now, let $i \in U$ be an element that appears in $2$ set in $\calS$. 
The subvectors of $\xx_i, \uu_{i,1}, \uu_{i,2}, \uu_{i,3}$ restricted to the coordinates in $B_i$ are: 
\[
\begin{pmatrix}
1 \\
0 \\
0 \\
0 \\
0
\end{pmatrix},
\begin{pmatrix}
1 \\
1 \\
1 \\
0 \\
0
\end{pmatrix},
\begin{pmatrix}
0 \\
1 \\
0 \\
1 \\
1
\end{pmatrix},
\text{ and }
\begin{pmatrix}
0 \\
0 \\
-1 \\
1 \\
1
\end{pmatrix}.
\]
We choose the signs in $\yy'$ for $\uu_{i,1}, \uu_{i,2}, \uu_{i,3}$ to be $-\yy(i), \yy(i), -\yy(i)$, respectively.
This guarantees $\sum_{i=1}^N \yy'(i) \xx_i(j) = 0$ for $j \in B_i$. 
Thus, the constructed $\yy'$ satisfies the conditions.
\end{proof}

\begin{proof}[Proof of Theorem \ref{thm:main1}]
Suppose the given $(3,2$-$2)$ \ssplit \  instance $\mathcal{I}(U,\calS)$ is satisfiable, meaning there exists an assignment $\yy \in \{\pm 1\}^{n}$ such that $\sum_{i=1}^n \yy(i)\xx_i(j) = 0$ for every $j \in \{1,\ldots,m\}$.
We construct a vector $\yy' \in \{\pm 1\}^{N}$ as in Claim \ref{clm:sign_u} that satisfies $\sum_{i=1}^N \yy'(i) \xx_i = {\bf 0}$.
We define a random vector $\zz \in \{\pm 1\}^{N}$ such that $\zz = \yy'$ with probability $1/2$ and $\zz = -\yy'$ with probability $1/2$.
Thus, $\mathbb{E}[\zz] = {\bf 0}$ and $\cov(\BB \zz) = {\bf 0}$. This implies $C(\BB, (1/2) \one) = 0$.

Next, suppose the given $(3,2$-$2)$ \ssplit \  instance $\mathcal{I}(U,\calS)$ is $\gamma$-unsatisfiable, meaning that for any assignment $\yy \in \{\pm 1\}^{n}$, at least $\gamma$ fraction of the entries of $\sum_{i=1}^n \yy(i) \xx_i$ are in $\{\pm 2, \pm 4\}$.
Then, for any $\yy' \in \{\pm 1\}^N$, at least $\frac{\gamma n}{N} \ge \frac{\gamma}{4}$ fraction of the entries of $\sum_{i=1}^N \yy'(i) \xx_i$ are in $\{\pm 2, \pm 4\}$.
Then, for any random $\zz \in \{\pm 1\}^{N}$ with $\mathbb{E}[\zz] = {\bf 0}$, 
let $\ww = \sum_{i=1}^N \zz(i) \xx_i$,
\begin{align*}
\norm{\cov \left( \ww  \right)}
= \norm{\mathbb{E} \left[ \ww \ww^\top
\right]}
\ge \frac{1}{d} \tr \left( \mathbb{E} \left[ \ww \ww^\top
\right] \right)  
= \frac{1}{d} \mathbb{E} \left[
\tr(\ww \ww^\top)
\right]
\ge \frac{1}{d} \cdot 4 \cdot \frac{\gamma N}{4}
= \frac{\gamma N}{d}
\ge \frac{4 \gamma}{23}.
\end{align*}
The last inequality holds since $d \le m+5n \le \frac{23n}{4} \le \frac{23N}{4}$.
That is, $C(\BB, (1/2) \one) > \frac{4\gamma}{23}$.
If we can distinguish whether $C(\BB, (1/2) \one) = 0$ or $C(\BB, (1/2) \one) > \frac{4\gamma}{23}$, 
then we can distinguish whether a $(3,2$-$2)$ \ssplit \ instance is  satisfiable
or $\gamma$-unsatisfiable, which is NP-hard by Theorem \ref{thm:322sspit_hard}.
\end{proof}

\subsection{Proof of Theorem \ref{thm:main2_appendix}}
\label{sec:thm2}

In this section, we prove Theorem \ref{thm:main2_appendix}.

Given a $(3,2$-$2)$ \ssplit \  instance $\mathcal{I}(U, \mathcal{S})$ where $\abs{U} = n$ and $\abs{\calS} = m$, we will construct a matrix $\BB$ and a probability vector $\pp$.
Let $\AA \in \{0,1\}^{m \times n}$ be the incidence matrix of the $(3,2$-$2)$ \ssplit \ instance $\mathcal{I}(U, \mathcal{S})$, where $\AA(j,i) = 1$ if element $i \in S_j$ and $\AA(j,i) = 0$ otherwise.
Since each set in $\mathcal{S}$ has $4$ distinct elements, each row of $\AA$ has a sum of value $4$.
We define a larger matrix:
\[
  \MM \defeq \begin{pmatrix}
    \AA & -2\II & -2\II \\
    {\bf 0} & \PPi & {\bf 0} \\
    {\bf 0} & {\bf 0} & \PPi
  \end{pmatrix} \in \mathbb{R}^{3m \times (n+2m)},
\]
where $\II \in \mathbb{R}^{m \times m}$ is the identity matrix and $\PPi$ is the orthogonal projection matrix onto the subspace of $\mathbb{R}^m$ that is orthogonal to the all-one vector.
Let 
\[
D = 3m, ~ N = n+2m.
\]
Observe that all columns of $\MM$ have Euclidean norm $\Theta(1)$.
Let 
\[
\BB = \frac{\MM}{\|\MM \|_{1,2}}.
\]
We will show that $\BB$ satisfies the conditions in Theorem \ref{thm:main2_appendix}.

Next, we construct the assignment probability vector $\pp$.
Let 
\[
p = 1-2\alpha \ge 0, ~ q = 2\beta - 1, ~ \lambda = (1-p) q.
\]
For a positive integer $k$, we let $\one_k$ be the all-one vector in $k$ dimensions.
We define
\[
\zz_0 \defeq \begin{pmatrix}
  p \one_{n} \\
  (p + \lambda) \one_m \\
  (p - \lambda) \one_m
\end{pmatrix} \in \mathbb{R}^{N},
\]
and
\[
\pp = \frac{\zz_0 + \one}{2} \in \{1-\alpha, 1- 2\alpha(1-\beta),  1- 2\alpha\beta\}^N.
\]

The following claim provides a simple formula for $\cov(\MM \zz)$ for $\zz$ with expectation $\zz_0$.

\begin{claim}
  If $\zz \in \mathbb{R}^N$ satisfies $\mathbb{E}[\zz] = \zz_0$, 
  then $\cov(\MM\zz) = \mathbb{E} \left[ \MM\zz \zz^\top \MM^\top \right]$.
  \label{clm:mx0}
\end{claim}
\begin{proof}
  Note that 
  \[
  \cov(\MM \zz) = \mathbb{E} \left[
    \MM (\zz - \zz_0) (\zz - \zz_0)^\top \MM^\top
  \right].  
  \]
  It suffices to show that $\MM \zz_0 = {\bf 0}$.
  By our construction of $\MM$,
  \begin{align*}
    \MM \zz_0 = \begin{pmatrix}
      p \AA \one_n  -2 (p + \lambda) \one_m - 2(p - \lambda) \one_m \\
      (p + \lambda) \PPi \one_m \\
      (p - \lambda) \PPi \one_m
    \end{pmatrix}.
  \end{align*}
  Since $\AA \one = 4 \one$ and $\PPi \one = {\bf 0}$,
  we have $\MM \zz_0 = {\bf 0}$.
\end{proof}

\subsubsection{Satisfiable $(3,2$-$2)$ \ssplit \  Instance}

\begin{lemma}
Suppose the $(3,2$-$2)$ \ssplit \  instance $\mathcal{I}(U,\mathcal{S})$ satisfiable.
Then, $C(\BB, \pp) = 0$, that is, there exists a random $\zz \in \{\pm 1\}^{N}$ such that 
$\mathbb{E}[\zz] = \zz_0$ and $\cov(\BB \zz) = {\bf 0}$.  
\label{lem:satisfiable}
\end{lemma}

We construct a random $\zz \in \{\pm 1\}^N$ as follows.
Let $\yy \in \{\pm 1\}^n$ be an assignment that splits all the sets in $\mathcal{S}$.
Then, $\AA \yy = {\bf 0}$.
Let
\[
\aa = \begin{pmatrix}
\one_m \\
- \one_m
\end{pmatrix}, 
\yy^{(1)} = \begin{pmatrix}
\yy \\
\aa 
\end{pmatrix},
\yy^{(2)} = \begin{pmatrix}
-\yy \\
\aa 
\end{pmatrix},
\]
and let
\[
p_1 = \frac{(1-p)(1+q)}{4}, ~ p_2 = \frac{(1-p)(1-q)}{4}.
\]
We construct the following random $\zz \in \{\pm 1\}^N$:
let $\zz = \one_N$ with probability (w.p.) $p$, 
$\zz = \yy^{(1)}$ w.p. $p_1$, $\zz = \yy^{(2)}$ w.p. $p_1$, $\zz = -\yy^{(1)}$ w.p. $p_2$, and $\zz = -\yy^{(2)}$ w.p. $p_2$. We can check that $\zz$ is well-defined: 
\[
p + 2p_1 + 2p_2 = 1.
\]

\begin{proof}[Proof of Lemma \ref{lem:satisfiable}]
  By our setting of $\zz$:
  \begin{align*}
    \mathbb{E}[\zz] & = p \one_N + 
    \frac{(1-p)(1+q)}{2} \begin{pmatrix}
      {\bf 0}_n \\
      \one_m \\
      - \one_m
    \end{pmatrix}
    + \frac{(1-p)(1-q)}{2} \begin{pmatrix}
      {\bf 0}_n \\
      -\one_m \\
      \one_m
    \end{pmatrix} 
   = p\one_N + (1-p)q \begin{pmatrix}
    {\bf 0}_n \\
    \one_m \\
    - \one_m
  \end{pmatrix}  = \zz_0.
  \end{align*}

By Claim \ref{clm:mx0}, 
\[
  \cov(\MM\zz) = \mathbb{E}\left[
    \MM \zz \zz^\top \MM^\top
  \right].  
\]
We will show that $\MM\zz' = {\bf 0}$ for every $\zz'$ in the support of $\zz$.
For $\zz' = \one_N$,
\[
\MM\zz' = \begin{pmatrix}
  \AA \one - 4 \one \\
  \PPi \one \\
  \PPi \one
\end{pmatrix}  = {\bf 0},
\]
where we use the fact $\AA \one = 4 \one$.
For $\zz' = \begin{pmatrix}
  \pm \yy \\
  \one \\
  -\one
\end{pmatrix}$, 
\[
  \MM\zz' = \begin{pmatrix}
    \pm \AA \yy - 2\one + 2\one \\
    \PPi \one \\
    - \PPi \one
  \end{pmatrix}  = {\bf 0},
\]
where we use the fact $\AA \yy = {\bf 0}$.
Similarly, for $\zz' = \begin{pmatrix}
  \pm \yy \\
  -\one \\
  \one
\end{pmatrix}$, $\MM\zz' = {\bf 0}$.
Therefore,
\[
\cov(\BB \zz) = \| \MM \|_{1,2}^{-2} \cdot \cov(\MM \zz) = {\bf 0}.
\]
\end{proof}

\subsubsection{Unsatisfiable $(3,2$-$2)$ \ssplit \  Instance}

\begin{lemma}
Suppose the $(3,2$-$2)$ \ssplit \  instance $\mathcal{I}(U,\mathcal{S})$ is \emph{$\gamma$-unsatisfiable},
that is, for any $\yy \in \{\pm 1\}^n$, at least $\gamma$ fraction of the entries of $\AA \yy$ are in $\{\pm 2, \pm 4\}$.
Then, $C(\BB, \pp) = \Omega(\lambda^2) = \Omega((1-p)^2 q^2)$, that is, for any random $\zz \in \{\pm 1\}^N$ satisfying $\mathbb{E}[\zz] = \zz_0$, we must have $\|\cov(\BB \zz)\| = \Omega(\lambda^2)$. 
\label{lem:unsatisfiable}    
\end{lemma}

Let $\zz \in \{\pm 1\}^N$ be a random vectors satisfying $\mathbb{E}[\zz] = \zz_0$.
We let $\zz_1 = \zz(1:n)$, $\zz_2 = \zz(n+1:n+m)$, and $\zz_3=\zz(n+m+1:n+2m)$; that is, $\zz_1$ contains the first $n$ entries of $\zz$, $\zz_2$ contains the next $m$ entries, and $\zz_3$ contains the last $m$ entries. 
Then,
\begin{align*}
  \MM \zz = \begin{pmatrix}
    \AA \zz_1 - 2 (\zz_2 + \zz_3) \\
    \PPi \zz_2 \\
    \PPi \zz_3
  \end{pmatrix}.
\end{align*}

It suffices to show that $\|\cov(\MM \zz)\| = \Omega(\lambda^2)$. The following claim decomposes $\norm{\cov(\MM\zz)}$ into three terms.

\begin{claim}
For any $\zz$ satisfying $\mathbb{E}[\zz] = \zz_0$, 
\[
\norm{\cov(\MM\zz)} \ge \frac{1}{D} \cdot
\max \left\{ 
\mathbb{E} \norm{\AA \zz_1 - 2 (\zz_2 + \zz_3)}^2,
\mathbb{E} \norm{\PPi \zz_2}^2,
\mathbb{E} \norm{\PPi \zz_3}^2
\right\}.
\]
\label{clm:cov}
\end{claim}
\begin{proof}
By Claim \ref{clm:mx0},
\begin{align*}
 D \cdot \norm{\cov(\MM\zz)} & = D \cdot \norm{\mathbb{E} \left[ \MM \zz \zz^\top \MM^\top \right]} \\
  & \ge \tr \left(
    \mathbb{E}  \left[ \MM \zz \zz^\top \MM^\top \right]
  \right) \\ 
  & = \mathbb{E} \left[
    \tr \left(
      \MM \zz \zz^\top \MM^\top
    \right)
  \right] 
  \tag{by the linearity of matrix trace} 
  \\
  & = \mathbb{E} \left[
      \norm{\MM \zz }^2
  \right] \\
  & = \left(
    \mathbb{E} \left[ \norm{\AA \zz_1 - 2 (\zz_2 + \zz_3)}^2 \right]
    + \mathbb{E} \left[ \norm{\PPi \zz_2}^2 \right]
    + \mathbb{E} \left[ \norm{\PPi \zz_3}^2 \right]
  \right) \\
  & \ge  \max \left\{ 
    \mathbb{E} \left[\norm{\AA \zz_1 - 2 (\zz_2 + \zz_3)}^2 \right],
    \mathbb{E} \left[ \norm{\PPi \zz_2}^2 \right],
    \mathbb{E} \left[ \norm{\PPi \zz_3}^2 \right]
    \right\}.
\end{align*}
\end{proof}

We will show that at least one of the three terms in Claim \ref{clm:cov} is sufficiently large.
We first look at the last two terms $\norm{\PPi \zz_2}^2$ and $\norm{\PPi \zz_3}^2$. 
Define 
\[
\alpha(\Tilde{\zz}) \defeq  \frac{\Tilde{\zz}^\top \one}{m}, ~ \forall \Tilde{\zz} \in \{\pm 1\}^m
\]

\begin{claim}
Let $k \in \{2,3\}$. If $\mathbb{E}[\alpha(\zz_k)^2] \le 1 - \frac{\gamma \lambda^2}{24}$, then 
$\mathbb{E} \norm{\PPi \zz_k}^2 = \Omega(D \gamma \lambda^2)$.
\label{clm:small_alpha_part}
\end{claim}

\begin{proof}
Note that 
\[
\norm{\PPi \zz_k}^2 = \norm{\zz_k - \alpha(\zz_k) \cdot \one}^2
= (1- \alpha(\zz_k)^2) m.
\]
Take expectation:
\begin{align*}
\mathbb{E} \norm{\PPi \zz_k}^2 
= (1-\mathbb{E}[\alpha(\zz_k)^2]) m
= \Omega(D \gamma \lambda^2),
\end{align*}
where the last equality holds since $D = \Theta(m)$.
\end{proof}

\begin{claim}
If $\mathbb{E}[\alpha(\zz_k)^2] > 1 - \frac{\gamma \lambda^2}{24}$ for each $k \in \{2,3\}$,
then $\mathbb{E} 
\norm{\AA \zz_1 - 2(\zz_2 + \zz_3)}^2 = \Omega(D \lambda^2)$.
\label{clm:large_alpha}
\end{claim}

The idea is to show that under the assumption of Claim \ref{clm:large_alpha}, with probability $\Omega(\lambda)$, a large fraction of the
entries of $\zz_2 + \zz_3$ are $0$.
Assuming this event holds,
\[
\norm{\AA \zz_1 - 2(\zz_2 + \zz_3)}^2 \approx \norm{\AA \zz_1}^2 = \Omega(m),
\]
where the last equality holds since the $(3,2$-$2)$ \ssplit \  instance is $\gamma$-unsatisfiable.
Thus, Claim \ref{clm:large_alpha} holds.

We will need the following properties about $\alpha =  \alpha(\zz_k)$ for $k \in \{2,3\}$.

\begin{claim}
Let $\alpha \in [-1,1]$ be a random variable.
Then, for any $\delta \in (0,1)$, 
$\Pr(\abs{\alpha} \le \delta) \le \frac{1 - \mathbb{E}[\alpha^2]}{1 - \delta^2}$.
\label{clm:alpha}
\end{claim}
\begin{proof}
Note that 
\begin{align*}
\mathbb{E}[\alpha^2] \le \Pr(\abs{\alpha} \le \delta) \cdot \delta^2 + 
1 - \Pr(\abs{\alpha} \le \delta).
\end{align*}
By rearranging the inequality above, we can show that the claim statement holds.
\end{proof}

\begin{claim}
Let $\alpha \in [-1,1]$ be a random variable. Then,
for any $\delta \in (0,1)$,
\begin{align*}
& \Pr(\alpha > \delta) \ge \frac{\mathbb{E}[\alpha] + \delta(1-2\Pr(\abs{\alpha} \le \delta))}{1+\delta}, \\
&  \Pr(\alpha < -\delta) \ge \frac{ - \mathbb{E}[\alpha] + \delta(1- \Pr(\abs{\alpha} \le \delta))}{1+\delta}.
\end{align*}
\label{clm:alpha1}
\end{claim}


\begin{proof}
We introduce some notations:
\[
\pi = \Pr(\abs{\alpha} \le \delta), 
~ \pi_+ = \Pr(\alpha > \delta),
~ \pi_- = \Pr(\alpha < -\delta).
\]
Let $\mathbb{I}_{\abs{\alpha} > \delta}$ (respectively, $\mathbb{I}_{\abs{\alpha} \le \delta}$) be the indicator for $\abs{\alpha} > \delta$ (respectively, $\abs{\alpha} \le \delta$).
Then,
\begin{align*}
\mathbb{E}\left[ \alpha \mathbb{I}_{\abs{\alpha} > \delta} \right]
= \mathbb{E}[\alpha] - \mathbb{E}\left[ \alpha \mathbb{I}_{\abs{\alpha} \le \delta} \right]
\ge \mathbb{E}[\alpha] - \delta \pi.
\end{align*}
In addition,
\begin{align*}
\mathbb{E}\left[ \alpha \mathbb{I}_{\abs{\alpha} > \delta} \right] 
\le \pi_+ - \delta \pi_-
= \pi_+ - \delta (1- \pi_+ - \pi)
= (1+\delta) \pi_+ + \delta \pi - \delta.
\end{align*}
Combining the above two inequalities obtains the first inequality in the claim statement.

To lower bound $\pi_-$, we note that 
\begin{align*}
\mathbb{E} \left[
  \alpha \mathbb{I}_{\abs{\alpha} > \delta} \right]
  \le \mathbb{E}[\alpha].
\end{align*}
In addition,
\begin{align*}
\mathbb{E} \left[
  \alpha \mathbb{I}_{\abs{\alpha} > \delta} \right]
  \ge \delta \pi_+ - \pi_-
  = \delta (1-\pi - \pi_-) - \pi_-
  = -(1+\delta) \pi_- - \delta\pi + \delta.
\end{align*}
Combining the above two inequalities obtains the second inequality in the claim statement.
\end{proof}

Now, we are ready to prove Claim \ref{clm:large_alpha}.
\begin{proof}[Proof of Claim \ref{clm:large_alpha}]

We choose $\delta = 1 - \frac{\gamma \lambda}{10} \in (0,1)$.
By our choice of $\zz_0$,
\begin{align}
\mathbb{E}[\alpha(\zz_2)] = p + \lambda, ~ 
\mathbb{E}[\alpha(\zz_3)] = p - \lambda.
\label{eqn:expect_alpha}
\end{align}
Let $\calE$ be the event that both $\alpha(\zz_2) > \delta
\text{ and } \alpha(\zz_3) < -\delta$ happen.
Let $\bar{\calE}$ be the complement of $\mathcal{E}$.
Then,
\begin{align*} 
\Pr \left( \calE \right) 
\ge & 1 - \Pr(\bar{\calE}) \\
\ge & 1 - \left(\Pr(\alpha(\zz_2) \le \delta) + \Pr(\alpha(\zz_3) \ge -\delta) \right) 
\tag{by a union bound} \\
= & \Pr\left( \alpha(\zz_2) > \delta \right)
+ \Pr \left( \alpha(\zz_3) < -\delta
\right)  - 1  \\
\ge & \frac{\mathbb{E}[\alpha(\zz_2)] - \mathbb{E}[\alpha(\zz_3)] + 2 \delta - 2 \delta \left(\Pr(\abs{\alpha(\zz_2)} \le \delta) + \Pr(\abs{\alpha(\zz_3)} \le \delta) \right)}{1 + \delta} - 1 
\tag{by Claim \ref{clm:alpha1}} \\
\ge & \frac{2\lambda + 2 \delta - \frac{2\delta}{1-\delta^2} (2 - \mathbb{E}[\alpha(\zz_2)^2] - \mathbb{E}[\alpha(\zz_3)^2]) }{1 + \delta} - 1  
\tag{by Equation \eqref{eqn:expect_alpha} and Claim \ref{clm:alpha}} \\
\ge & \frac{1}{1+\delta} \left(
2 \lambda - \frac{\gamma\lambda}{10} -  \left( 1 - \frac{\gamma\lambda}{10} \right) \left(\frac{ \gamma \lambda^2 / 6}{1-\left(1 - \gamma \lambda / 10\right)^2}\right)  \right) 
\tag{by our setting of $\delta$ and assumption on $\mathbb{E}[\alpha(\zz_2)^2], \mathbb{E}[\alpha(\zz_3)^2]$} \\
\ge & \frac{1}{1+\delta} \left(
2 \lambda  - \frac{\gamma\lambda}{10} -  \left( 1 - \frac{\gamma\lambda}{10} \right) \lambda 
\right) 
\tag{since $\gamma \lambda \le 1$}
\\
\ge & \frac{\lambda}{2} \left(1 - \frac{\gamma}{10}\right). 
\tag{since $1+\delta \le 2$}
\end{align*}
Assuming event $\calE$ happens, at least 
$$\frac{1+\alpha(\zz_2)}{2} > 1 - \frac{\gamma \lambda}{20}$$
fraction of the entries of $\zz_2$ are $1$, and at least 
\[
\frac{1-\alpha(\zz_3)}{2} > 1 - \frac{\gamma \lambda}{20}
\]
fraction of the entries of $\zz_3$ are $-1$.
Thus, at least $1 - \frac{\gamma \lambda}{10}$ fraction of the entries of $\zz_2 + \zz_3$ are $0$.
Among these $0$-valued entries of $\zz_2 + \zz_3$, at least $\gamma(1-\frac{\lambda}{10})$ fraction of the entries of $\AA \zz_1$ are in $\{\pm 2, \pm 4\}$.
In this case,
\begin{align*}
\norm{\AA \zz_1 - 2(\zz_2 + \zz_3)}^2
\ge 4 \gamma \left(1 - \frac{\lambda}{10}\right) m.
\end{align*}
Therefore,
\begin{align*}
\mathbb{E} \left[
\norm{\AA \zz_1 - 2(\zz_2 + \zz_3)}^2
\right] 
& \ge \Pr(\calE) \cdot 4 \gamma \left(1 - \frac{\lambda}{10}\right) m \\
& \ge \frac{\lambda}{2} \left(
1 - \frac{\gamma}{10}
\right) \cdot 4 \gamma \left(1 - \frac{\lambda}{10}\right) m \\
& = \Omega(D \lambda) 
\tag{since $D = \Omega(m)$}
\end{align*}
\end{proof}

Lemma \ref{lem:unsatisfiable} follows Claims \ref{clm:cov}, \ref{clm:small_alpha_part}, and \ref{clm:large_alpha}.

Theorem \ref{thm:main2_appendix} is derived from Lemmas \ref{lem:satisfiable} and \ref{lem:unsatisfiable} and Theorem \ref{thm:322sspit_hard}.

%% file: appendix_algo.tex
\section{Missing Proofs in Section \ref{sec:algo}}
\label{sec:appendix_algo_proofs}

In this section, we prove Theorems \ref{thm:mwu} and \ref{thm:oracle_algo} in Section \ref{sec:algo}.

\subsection{Proofs for Theorem \ref{thm:mwu}}
\label{sec:appendix_mwu_known_cov}

We start with presenting a proof that assumes the matrices $\cov(\BB \zz_t)$, which appear in lines \ref{lin:alphat} and \ref{lin:wt} of Algorithm \ref{alg:mwu}, have explicit forms (that is, given $\BB, \WW_{t-1}, \pp, \eps$, an oracle can return an explicit form of $\cov(\BB \zz_t)$). This assumption simplifies our proof: Under it, all covariance matrices $\cov(\BB \zz_t)$, weight matrices $\WW_t$, probabilities $p_t$, and the number of while-iterations are deterministic.
Later, in Section \ref{sec:appendix_unknown_covariance}, we explain how to drop this assumption by estimating the covariance matrices using their empirical means.

Let $\zz \leftarrow \textsc{MWU}(\BB, \pp, \mathcal{O}, \eta, \eps)$ returned by Algorithm \ref{alg:mwu}. 
Let $T$ be the last while-iteration, that is, $\zz = \zz_T$.
For $t=1,\ldots,T$, let $p_t \defeq \alpha_t / \alpha$.

By our assumption of the oracle that $\mathbb{E}[\zz_{t}] = 2\pp - \one$ for $t=1,\ldots,T$, we have the following lemma.

\begin{lemma}
Assuming all $\cov(\BB \zz_t)$ have known forms in Algorithm \ref{alg:mwu}, we have $\mathbb{E}[\zz] = 2\pp - \one$.
\label{lem:mwu_expect}
\end{lemma}
\begin{proof}
Let $\zz_0 = 2\pp - \one$. 
Under the assumption, probabilities $p_t$ and iterations $T$ are deterministic.
By our assumption on the oracle $\mathcal{O}$, we have $\mathbb{E}[\zz_t] = \zz_0$ for each $t=1,\ldots,T$. Thus,
\begin{align*}
\mathbb{E}[\zz] = \mathbb{E} \left[ \sum_{t=1}^T p_t \zz_t \right] 
= \zz_0.
\end{align*}    
\end{proof}

It remains to provide an upper bound for $\|\cov(\BB \zz) \|$.
We will need the following facts on matrix exponential and matrix trace.
For two symmetric matrices $\AA, \BB \in \mathbb{R}^{m \times m}$, $\AA \pleq \BB$ is the Loewner order meaning that $\BB - \AA$ is positive semidefinite.

\begin{theorem}[Golden–Thompson inequality \citep{golden1965lower,thompson1965inequality}]
Let $\AA, \BB \in \mathbb{R}^{m \times m}$ be two symmetric matrices. Then, $\tr(e^{\AA + \BB}) \le \tr(e^{\AA} e^{\BB})$.
\label{thm:golden_thompson}
\end{theorem}

\begin{lemma}
Let $\AA, \BB, \CC \in \mathbb{R}^{m \times m}$ be symmetric positive semidefinite matrices such that $\AA \pleq \BB$. Then, $\tr(\AA \CC) \le \tr(\BB \CC)$.
\label{lem:trace}
\end{lemma}
\begin{proof}
Let $\CC^{1/2}$ be the square root of matrix $\CC$, that is, $\CC^{1/2}$ is symmetric positive semidefinite and $\CC^{1/2} \CC^{1/2} = \CC$. Since $\AA \pleq \BB$, we have 
\[
\CC^{1/2} \AA \CC^{1/2} \pleq \CC^{1/2} \BB \CC^{1/2}
\implies \tr(\CC^{1/2} \AA \CC^{1/2}) \le \tr(\CC^{1/2} \BB \CC^{1/2}) \implies \tr(\AA \CC) \le \tr(\BB \CC).
\]
\end{proof}

\begin{claim}
Let $\VV \in \mathbb{R}^{m \times n}$ where each column has unit norm, and let $\zz_0 \in [-1,1]^n$. Then, 
\[
\max_{\zz \in \{\pm 1\}^n} \norm{\VV (\zz - \zz_0) (\zz - \zz_0)^\top \VV^\top} \le 4n^2.
\]
\label{clm:Vx_norm}
\end{claim}

\begin{proof}
\[
\max_{\zz \in \{\pm 1\}^n} \norm{\VV (\zz - \zz_0) (\zz - \zz_0)^\top \VV^\top}
\le \norm{\VV}^2 \max_{\zz \in \{\pm 1\}^n} \norm{\zz - \zz_0}^2 \le 4n^2.
\]
\end{proof}

\begin{proof}[Proof of Theorem \ref{thm:mwu} assuming known forms for $\cov(\BB \zz_{t})$ in Algorithm \ref{alg:mwu}]
By line \ref{lin:sample} of Algorithm \ref{alg:mwu}, we have 
\[
\cov(\BB \zz) = \sum_{t=1}^T p_t \cov(\BB \zz_t).
\]
Let $\MM_t \defeq \cov(\BB \zz_t)$.
For each $t=0,1, \ldots, T$, we define
a potential function $\Phi_t \defeq \tr(\WW_t)$ where $\WW_t$ is defined in line \ref{lin:wt}. 
Then,
\begin{align}
\Phi_t 
& = \tr \exp \left(\sum_{\tau=1}^t \alpha_{\tau} \MM_{\tau} \right) \nonumber \\
& \le \tr \left( \exp \left(  \sum_{\tau=1}^{t-1}  
\alpha_{\tau} \MM_{\tau} \right) 
\exp \left(\alpha_t \MM_t \right)
\right) 
\tag{by the Golden-Thompson inequality (Theorem \ref{thm:golden_thompson})} 
\\
& \le \tr\left( \exp \left( \sum_{\tau=1}^{t-1} 
\alpha_{\tau} \MM_{\tau} \right)  \left( \II + (1+\eps/6)\alpha_t \MM_t \right) \right)
\tag{since $e^{\AA} \pleq \II + \AA + \AA^2$ for $\AA \pleq 2 \II$ and Lemma \ref{lem:trace}} 
\\
& = \Phi_{t-1} + (1+\eps/6) \alpha_t \cdot \tr(\WW_{t-1} \MM_t) \label{eqn:phi_t}  \\
& \le \Phi_{t-1} + (1+\eps/6) \alpha_t \eta \cdot \tr(\WW_{t-1}) 
\tag{by the assumption of oracle $\mathcal{O}$}\\
& = (1 + (1+\eps/6) \alpha_t \eta) \Phi_{t-1}. \nonumber 
\end{align}
Recursively applying the above inequality, we have
\begin{align*}
\Phi_T &\le \prod_{t=1}^T (1 + (1+\eps/6) \alpha_t \eta ) \Phi_0  \\
& \le \exp \left( (1+\eps/6) \eta \sum_{t=1}^T \alpha_t \right) \Phi_0 \\
& \le \exp \left((1+\eps/6)\eta \alpha \right) m 
\tag{since $\Phi_0 = m$}
\end{align*}
By the definition of $\Phi_T$ and $\WW_T$,
\begin{align*}
\norm{\exp\left(\sum_{t=1}^T \alpha_t \MM_t \right)} \le \tr\left( \exp\left(\sum_{t=1}^T \alpha_t \MM_t \right) \right) \le \exp \left((1+\eps/6)\eta \alpha \right) m.
\end{align*}
Take the logarithm on both sides, 
\begin{align*}
\norm{\sum_{t=1}^T \alpha_t \MM_t }
\le (1+\eps/6)\eta  \alpha + \ln m.
\end{align*}
Diving $\alpha$ on both sides,
\[
\norm{\sum_{t=1}^T p_t \MM_t }
\le (1+\eps/6) \eta + \frac{\ln m}{\alpha}
\le (1+\eps) \eta,
\]
where the second inequality follows the termination criteria in line \ref{lin:while}.
Thus, 
$\norm{\cov(\BB \zz)} \le (1+\eps) \eta$.

We next bound the runtime of Algorithm \ref{alg:mwu}.
Since 
\begin{align*}
\norm{\MM_t} \le \max_{\zz+\zz_0 \in \{\pm 1\}^n} \norm{\BB \zz}_2^2 
\le 4n^2,
\end{align*}
where the second inequality follows Claim \ref{clm:Vx_norm},
we have $\alpha_t \ge \frac{\eps}{24n^2}$ for each $t=1,\ldots,T$. Thus, the total number of iterations in the while-loop is at most $\frac{48n^2 \ln m}{\eps^2 \eta}$.
So, the number of calls to $\mathcal{O}$ and the algorithm's runtime are polynomial.
\end{proof}

\subsubsection{Unknown Covariance Matrices $\cov(\BB \zz_{t})$}
\label{sec:appendix_unknown_covariance}

When covariance matrices $\cov(\BB \zz_{t})$ in Algorithm \ref{alg:mwu} do not have explicit forms, we can estimate them by their empirical means.
Specifically, we replace each $\MM_t \defeq \cov(\BB \zz_t)$ by $\widetilde{\MM}_t$ defined as follows:
We pre-specify a parameter $N$. 
For $k = 1,\ldots,N$, independently sample $\zz_{t,k}\leftarrow \mathcal{O}(\BB, \WW_{t-1}, \pp)$.
Let  
\[
\widetilde{\MM}_{t} = \frac{1}{N} \sum_{k=1}^N \BB (\zz_{t,k} - \zz_0) (\zz_{t,k} - \zz_0)^\top \BB^\top,
\]
where $\zz_0 = 2\pp - \one$.

Let $\zz$ be the output of Algorithm \ref{alg:mwu} using $\widetilde{\MM}_t$'s. Let $T$ be the last iteration, that is, $\zz = \zz_T$.
In addition, let $\mathbb{E}_t$ be the expectation conditioned on the first $t$ iterations.

\begin{lemma}
Assuming that we estimate each $\cov(\BB \zz_t)$ in Algorithm \ref{alg:mwu} using its empirical mean of $N$ independent samples, we have $\mathbb{E}[\zz] = 2\pp - \one$.  
\end{lemma}

\begin{proof}
Let $T_{\max}$ be a fixed upper bound of $T$.
For $T < t \le T_{\max}$, we let $p_t = 0$ and $\zz_t = {\bf 0}$.
For each $t$, conditioning on the first $t-1$ iterations, $p_t$ and $\zz_t$ are independent, and thus by the assumption of the oracle,
\[
\mathbb{E}_{t-1} [p_t \zz_t] = \mathbb{E}_{t-1} [p_t] \zz_0.
\]
Then,
\begin{align*}
\mathbb{E}[\zz] = \mathbb{E}\left[ \sum_{t=1}^{T_{\max}} p_t \zz_t \right]
= \mathbb{E} \left[\sum_{t=1}^{T_{\max}} p_t \right] \zz_0 = \zz_0.
\end{align*}
\end{proof}

For each $t=1,\ldots,T$, let $\cov_{\WW_{t-1}}(\BB \zz_t)$ be the covariance matrix where $\zz_t \leftarrow \mathcal{O}(\BB, \WW_{t-1}, \pp)$.
Suppose the algorithm produces a sequence $\alpha_1, \ldots, \alpha_T$ and a sequence $\WW_1, \ldots, \WW_{T-1}$ (which are no longer deterministic since they depend on random matrices $\widetilde{\MM}_1, \ldots, \widetilde{\MM}_{T-1}$).
We can express the covariance matrix $\cov(\BB \zz)$ as follows:
\begin{align}
\cov(\BB \zz) = \mathbb{E}_{\alpha_1, \ldots, \alpha_T, \WW_1, \ldots, \WW_{T-1}} \left[ \sum_{t=1}^T p_t \cov_{\WW_{t-1}}(\BB \zz_t)  \right]
= \mathbb{E}_{\widetilde{\MM}_1, \ldots, \widetilde{\MM}_{T-1}} \left[ \sum_{t=1}^T p_t \cov_{\WW_{t-1}}(\BB \zz_t)  \right].
\label{eqn:mwu_cov}
\end{align}
For notation brevity, we will drop the subscript $\alpha_1, \ldots, \alpha_T, \WW_1, \ldots, \WW_{T-1}$ or $\widetilde{\MM}_1, \ldots, \widetilde{\MM}_{T-1}$ when the context is clear.

We claim that if $N$ is chosen sufficiently large, given $\WW_{t-1}$, our estimate $\widetilde{\MM}_t$ is sufficiently close to the true value of $\MM_t$.
We will need the following theorem on approximating a covariance matrix by an empirical estimator.

\begin{theorem}[Covariance estimator (Ref: Section 1.6.3. of \cite{tropp2015introduction})]
Let $\yy \in \mathbb{R}^m$ be a random vector with $\mathbb{E}[\yy] = {\bf 0}$ and $\norm{\yy}^2 \le B$. Let $\cov(\yy) \defeq \mathbb{E}[\yy \yy^\top]$ be the covariance matrix of $\yy$.
Let $\yy_1, \ldots, \yy_N$ be $N$ independent copies of $\yy$, and let 
\[
\YY \defeq \frac{1}{N} \sum_{k=1}^N \yy_k \yy_k^\top.
\]
Then, 
\[
\Pr \left( \norm{\YY - \cov(\yy)} \ge \delta \right)
\le 2 m \exp \left( - \frac{3\delta^2 N}{12B^2 + 4B\delta} \right), ~ \forall \delta \ge 0
\]
\label{thm:cov_est}
\end{theorem}

\begin{proof}[Proof of Theorem \ref{thm:mwu} assuming no known forms for $\cov(\XX, \zz_{t})$ in Algorithm \ref{alg:mwu}]
We first bound the runtime of the algorithm. Let 
\begin{align}
\delta \defeq \frac{\eps \eta}{12 + \eps}.
\label{eqn:delta}
\end{align}
We choose 
\begin{align}
N = \lceil 100n^4 \delta^{-2}  \ln \left(10^4 n^4 m \eps^{-3} \eta^{-2} \ln m \right) \rceil,
\label{eqn:N}
\end{align}
which is polynomial in $n,m,(\eps \eta)^{-1}$.
Similar to the proof of Theorem \ref{thm:mwu} assuming known $\cov(\BB \zz_t)$ in Section \ref{sec:appendix_mwu_known_cov}, we can bound the number of iterations $T \le \frac{48n^2 \ln m}{\eps^2 \eta}$.
So, the algorithm has a polynomial runtime.

Next, we upper bound $\| \cov(\BB \zz) \|$.

Given any fixed $\WW_{t-1}$ and $\zz_t \leftarrow \mathcal{O}(\BB, \WW_{t-1}, \pp)$,
we apply Theorem \ref{thm:cov_est} with $\yy = \BB(\zz_t - \zz_0)$. 
Then, $\mathbb{E}[\yy] = {\bf 0}$ and $\norm{\yy}^2 \le 4n^2$ (by Claim \ref{clm:Vx_norm}). Let $\Pr_{\WW_{t-1}}$ be the probability conditioned on $\WW_{t-1}$. Then,
\begin{align*}
\Pr_{\WW_{t-1}} \left( \norm{\widetilde{\MM}_t - \cov_{\WW_{t-1}} (\BB \zz_t)} \ge \delta  \right)
\le 2d \exp \left( - \frac{3\delta^2 N}{192n^4 + 16n^2 \delta} \right) \defeq p.
\end{align*}
Let $\mathcal{E}$ be the event there exists $t \in \{1,\ldots,T\}$ such that conditioning on $\WW_{t-1}$, $\norm{\widetilde{\MM}_t - \cov_{\WW_{t-1}} (\BB \zz_t)} \ge \delta$.
We note $T \le \frac{48n^2 \ln m}{\eps^2 \eta} \defeq T_u$. 
Then,
\[
\Pr \left( \mathcal{E} \right)
\le p T_u.
\]

We condition on event $\Bar{\mathcal{E}}$, the complement of event $\mathcal{E}$.
For each $t=0,1, \ldots, T$, we define
a potential function $\Phi_t \defeq \tr(\WW_t)$, and we reload the notation $\MM_t \defeq \cov_{\WW_{t-1}} (\BB \zz_t)$. 
By Equation \eqref{eqn:phi_t} (replacing $\MM_{\tau}$ with $\widetilde{\MM}_{\tau}$ for $\tau=1,\ldots,t$), we have 
\[
\Phi_t  \le \Phi_{t-1} + (1+\eps/6) \alpha_t \cdot \tr(\WW_{t-1} \widetilde{\MM}_t).
\]
Since $\norm{\widetilde{\MM}_t - \MM_t} \le \delta$ and we assume $\tr(\WW_{t-1} \MM_t) \le \eta \tr(\WW_{t-1})$, we have 
\[
\tr(\WW_{t-1} \widetilde{\MM}_t) \le \tr(\WW_{t-1} \MM_t) + \delta \tr(\WW_{t-1})
\le (\eta + \delta) \tr(\WW_{t-1}).
\]
Plugging into the equation on $\Phi_t$ and $\Phi_{t-1}$: 
\begin{align*}
\Phi_t \le \Phi_{t-1} + (1+\eps/6) \alpha_t (\eta +\delta) \tr(\WW_{t-1})
= \left(1 + (1+\eps/6) \alpha_t (\eta +\delta) \right) \Phi_{t-1}.
\end{align*}
Following the rest of the proof of Theorem \ref{thm:mwu} in Section \ref{sec:appendix_mwu_known_cov} where we replace $\eta$ with $\eta + \delta$ and replace $\MM_t$ with $\widetilde{\MM}_t$, we have 
\[
\norm{\sum_{t=1}^T p_t \widetilde{\MM}_t }
\le (1+\eps/6) (\eta+\delta) + \frac{\ln m}{\alpha}
\le (1+2\eps/3) \eta + \frac{\eps \delta}{6}.
\]
Thus, 
\begin{align*}
\norm{\sum_{t=1}^T p_t \MM_t}
& \le \norm{\sum_{t=1}^T p_t \widetilde{\MM}_t}
+ \norm{\sum_{t=1}^T p_t (\widetilde{\MM}_t - \MM_t)} 
\le (1+2\eps/3) \eta + (2+\eps/6) \delta
\le (1+5\eps/6) \eta,
\end{align*}
where the last inequality is by our setting of $\delta$ in Equation \eqref{eqn:delta}.

In addition, we have (no matter $\mathcal{E}$ happens or not)
\begin{align*}
\norm{\sum_{t=1}^T p_t \cov_{\WW_{t-1}}(\BB \zz_t) } \le \max_{\zz+\zz_0 \in \{\pm 1\}^n} \norm{\BB \zz \zz^\top \BB}
\le 4n^2,
\end{align*}
where the second inequality follows Claim \ref{clm:Vx_norm}.

By Equation \eqref{eqn:mwu_cov},
\begin{align*}  
\norm{\cov(\BB \zz)}
& \le \mathbb{E}_{\widetilde{\MM}_1, \ldots, \widetilde{\MM}_{T-1}} \norm{\sum_{t=1}^T p_t \cov_{\WW_{t-1}}(\BB \zz_t)} \\
& = \Pr(\Bar{\mathcal{E}}) \mathbb{E} \left[ \left. \norm{\sum_{t=1}^T p_t \cov_{\WW_{t-1}}(\BB \zz_t)} \right| \Bar{\mathcal{E}} \right] 
+ \Pr(\mathcal{E}) \mathbb{E} \left[ \left. \norm{\sum_{t=1}^T p_t \cov_{\WW_{t-1}}(\BB \zz_t)} \right| \mathcal{E} \right] \\
& \le (1+5\eps/6)\eta + p T_u \cdot 4n^2 \\
& \le (1+\eps) \eta,
\end{align*}
where the last inequality holds due to our choice of $N$.
\end{proof}

\paragraph{Remark on runtime.}
Our choice of $N$ in Equation \eqref{eqn:N} guarantees the stated covariance norm bound in the worst case. In practice, a much smaller $N$ such as $N = O(m \log m)$ might be sufficient.
Our algorithm is slower than the GSW design in \cite{HSSZ19}. 
However, our algorithm's runtime remains \emph{polynomial} in the size of the inputs, which is acceptable for a broad spectrum of randomized experiments. 
In the fields of medicine, agriculture, and education, many experiments have no more than hundreds of covariates and units; running our algorithm takes a reasonable amount of time (depending on the parameters).
The additional time required for planning/designing a randomized experiment may be negligible in comparison to the potentially years-long duration of the experiment’s implementation. In such cases, the increased computational cost of planning is far outweighed by the gains in estimation accuracy.

\subsection{Proofs for Theorem \ref{thm:oracle_algo}}


Let $\zz \leftarrow \textsc{Oracle}(\BB, \WW, \pp, \eps)$. 
Let $T$ be the last while-iteration, that is, $\zz = \zz_T$.
Our choice of $\gamma_t$ guarantees $\mathbb{E}[\zz_t - \zz_{t-1}] = {\bf 0}$ for each $t$, 
which leads to the following lemma.

\begin{lemma}
$\mathbb{E}[\zz] = \zz_0 = 2\pp - \one$.
\label{lem:oracle_expect}
\end{lemma}
\begin{proof}
By our choice of $\gamma_t$, for each $t=1,\ldots,T$, $\mathbb{E}_{t-1} [\zz_t - \zz_{t-1}] = {\bf 0}$. Thus,
\begin{align*}
\mathbb{E}[\zz] = \mathbb{E}\left[ \zz_0 + \sum_{t=1}^T (\zz_t - \zz_{t-1}) \right]  = \zz_0.
\end{align*}
\end{proof}

We upper bound $\langle \cov(\BB \zz), \WW \rangle$. 
We need the following property about the update vector $\yy_t$.

Let $\DD_t \defeq \diag(\VV_{t,l}^\top \WW_t \VV_{t,l})$.
Let $\Tilde{\yy}_t = \yy_t(A_t)$ be the vector in Equation \eqref{eqn:y_t}.

\begin{claim}
$\tilde{\yy}_t^\top \MM_t \tilde{\yy}_t \ge 0$.
\label{clm:yt}
\end{claim}

\begin{proof}
We show that there exists a linear subspace $\mathcal{U} \subset \mathbb{R}^{A_t}$ of dimension at least $\lceil \frac{\eps}{1+\eps} \abs{A_t} \rceil$ such that 
\begin{align}
\forall \uu \in \mathcal{U}, ~ \uu^\top \MM_t \uu \ge 0.
\label{eqn:U_goal}
\end{align}
By the definition of $\MM_t$,
\[
\MM_t = (1+\eps) \DD_t - \VV_{t,l}^\top \WW_t \VV_{t,l}.
\]
Since $\VV_{t,l}^\top \WW_t \VV_{t,l}$ is positive semidefinite, if its $i$th diagonal is zero, then its $i$th row and $i$th column are all zero. 
Let $\DD_t^{1/2}$ be the square root of $\DD_t$ and $\DD_t^{\dagger 1/2}$ be the pseudo-inverse of $\DD_t^{1/2}$.
Then,
\[
\MM_t = \DD_t^{1/2} \left( (1+\eps) \II - 
\DD_t^{\dagger 1/2} \VV_{t,l}^\top \WW_t \VV_{t,l} \DD_t^{\dagger 1/2} \right) \DD_t^{1/2}. 
\]
Let $\lambda_1, \ldots, \lambda_{\abs{A_t}} \ge 0$ be the eigenvalues of $\DD_t^{\dagger 1/2} \VV_{t,l}^\top \WW_t \VV_{t,l} \DD_t^{\dagger 1/2}$. Then, 
\[
\abs{A_t} \ge \sum_i \lambda_i > (1+\eps) \cdot \abs{\{i: \lambda_i > 1+\eps \}}.
\]
Rearranging the above inequality:
\[
\abs{\{i: \lambda_i \le 1+\eps \}} \ge \lceil \frac{\eps}{1+\eps} \abs{A_t} \rceil.
\]
We let $\mathcal{U}$ be the linear subspace spanned by the eigenvectors associated with $\lambda_i \le 1+\eps$. Then, $\mathcal{U}$ satisfies Equation \eqref{eqn:U_goal}.

In addition, the rank of the subspace that is orthogonal to $\VV_{t,b}$ is at least 
\[
\abs{A_t} - \abs{B_t} > \abs{A_t} - \frac{\eps}{1+\eps} \abs{A_t} = \frac{1}{1+\eps} \abs{A_t}.
\]
Thus, there exists $\tilde{\yy} \in \mathcal{U}$ that is orthogonal to $\VV_{t,b}$ and satisfies $\tilde{\yy}^\top \MM_t \tilde{\yy} \ge 0$. Therefore, $\tilde{\yy}_t^\top \MM_t \tilde{\yy}_t \ge 0$.
\end{proof}

For each $t=0,1,\ldots$, let
\begin{align}
D_t \defeq \langle \BB (\zz_t - \zz_0) (\zz_t - \zz_0)^\top \BB^\top, \WW \rangle.
\label{eqn:Dt_def}
\end{align}
Note that $D_0 = 0$ and $\mathbb{E}[D_T] = \langle \cov(\BB \zz), \WW \rangle$.

\begin{claim}
$\mathbb{E}_{t-1}[D_t - D_{t-1}] \le (1+\eps) \cdot \mathbb{E}_{t-1} [\gamma_t^2 \cdot \tilde{\yy}_t^\top \DD_t \tilde{\yy}_t]$.
\label{clm:Dt}
\end{claim}

\begin{proof}
By definition of $D_t$ and the fact that $\mathbb{E}_{t-1}[\zz_t] = \zz_{t-1}$,
\[
\mathbb{E}_{t-1} [D_t - D_{t-1}] 
=  \left\langle \BB \WW \BB^\top, \mathbb{E}_{t-1} \left[ \zz_t \zz_t^\top - \zz_{t-1} \zz_{t-1}^\top \right] \right\rangle. 
\]
By our update of $\zz_t$, for each $i,j \in A_t$,
\begin{align*}
\mathbb{E}_{t-1} [\zz_t(i) \zz_t(j) - \zz_{t-1}(i) \zz_{t-1}(j)]
& = \mathbb{E}_{t-1} \left[ (\zz_{t-1}(i) + \gamma_t \yy_t(i)) (\zz_{t-1}(j) + \gamma_t \yy_t(j)) - \zz_{t-1}(i) \zz_{t-1}(j) \right] \\
& = \mathbb{E}_{t-1} [\gamma_t^2 \yy_t(i) \yy_t(j) ]
\tag{since $\mathbb{E}_{t-1} [\gamma_t] = 0$}
\end{align*}
For each $i,j$ where at least one of $i,j$ not in $A_t$,
\[
\mathbb{E}_{t-1} [\zz_t(i) \zz_t(j) - \zz_{t-1}(i) \zz_{t-1}(j)] = 0.
\]
Thus, 
\begin{align*}
\mathbb{E}_{t-1} [D_t - D_{t-1}] 
= & \mathbb{E}_{t-1} \left[ \gamma_t^2 \left\langle \widetilde{\BB}_t^\top \WW \widetilde{\BB}_t, \yy_t(A_t) \yy_t(A_t)^\top \right\rangle
\right] \\
= & \mathbb{E}_{t-1} \left[ \gamma_t^2 \left\langle \VV_{t,l}^\top \WW_t \VV_{t,l}, \yy_t(A_t) \yy_t(A_t)^\top \right\rangle
\right] 
\tag{since $\VV_{t,b} \yy_t(A_t) = {\bf 0}$} 
\\
\le & (1+\eps) \mathbb{E}_{t-1} [\gamma_t^2 \cdot \tilde{\yy}_t^\top \DD_t \tilde{\yy}_t ]
\tag{by Claim \ref{clm:yt}} 
\end{align*}
\end{proof}

Below, we bound $\langle \cov(\BB \zz), \WW \rangle$ in terms of $U_W = \langle \cov(\BB \zz_{Bernoulli}), \WW \rangle$.

\begin{lemma}
$\langle \cov(\BB \zz), \WW \rangle \le (1+\eps) U_W$.
\label{lem:tr_bound_uniform}
\end{lemma}

\begin{proof}
Without loss of generality, assume $\WW = \diag(\lambda_1, \ldots, \lambda_m)$ (otherwise, we can do a linear transform). Let $\bb_1, \ldots, \bb_n$ be the columns of $\BB$.
Then,
\[
U_W = \sum_{j=1}^m \lambda_j \sum_{i=1}^n \bb_i(j)^2 (1-\zz_0(i)^2).
\]
Let $\DD \defeq \diag(\bb_i^\top \WW \bb_i, i=1,\ldots, n)$. 
For each $t=0,1,\ldots$, let $E_t \defeq \zz_t^\top \DD \zz_t$.
\begin{align}
\mathbb{E}_{t-1} [E_t - E_{t-1}]
& = \mathbb{E}_{t-1} \left[ \left\langle \DD, \zz_t \zz_t^\top - \zz_{t-1} \zz_{t-1}^\top \right\rangle
\right] \nonumber \\
& = \mathbb{E}_{t-1} [ \gamma_t^2  \cdot \yy_t(A_t)^\top \DD(A_t, A_t) \yy_t(A_t) ]
\label{eqn:et_diff} \\
& \ge \mathbb{E}_{t-1} [ \gamma_t^2  \cdot \yy_t(A_t)^\top \DD_t \yy_t(A_t) ]
\tag{since $\DD(A_t, A_t) \pgeq \DD_t$} \\
& \ge (1+\eps)^{-1} \mathbb{E}_{t-1} [D_t - D_{t-1}]
\tag{by Claim \ref{clm:Dt}}
\end{align}
Sum up the above equation over $t=1,\ldots,T$,
\begin{align*}
\mathbb{E}[D_T] & \le (1+\eps) \cdot \left( \mathbb{E}[E_T] -  E_0 \right) \\ 
& = (1+\eps) \left( \zz_T^\top \DD \zz_T - \zz_0^\top \DD \zz_0 \right) \\
& = (1+\eps) \sum_{i=1}^n \bb_i^\top \WW \bb_i \left(1 - \zz_0(i)^2 \right) 
\tag{since $\zz_T \in \{\pm 1\}^n$} \\
& = (1+\eps) \sum_{i=1}^n \sum_{j=1}^m \lambda_j \bb_i(j)^2  \left(1 - \zz_0(i)^2 \right) \\
& = (1+\eps) U_W.
\end{align*}
\end{proof}

Finally, we bound $\langle \cov(\BB \zz), \WW \rangle$ in terms of $\langle \cov(\BB \zz_{GSW}), \WW \rangle$.

\begin{lemma}
$\langle \cov(\BB \zz), \WW \rangle \le (1+\eps) (1 + 1/\eps)$.
\label{lem:tr_bound_const}
\end{lemma}

\begin{proof}
For each $j=1,\ldots,m$, we let $t_j \defeq \min\{t: j \in L_t\}$ and let 
\[
E_t(j) \defeq \sum_{i \in A_{t_j}} \zz_t(i)^2 \bb_i(j)^2, ~ \forall t \ge t_j
\]
Then, for each $j$ and $t > t_j$,
\begin{align*}
\mathbb{E}_{t-1} [E_t(j) - E_{t-1}(j)]
& = \mathbb{E}_{t-1} \left[ \sum_{i \in A_{t_j}} (\zz_t(i)^2 - \zz_{t-1}(i)^2) \bb_i(j)^2 \right] \\
& = \mathbb{E}_{t-1} \left[ \sum_{i \in A_{t}} (\zz_t(i)^2 - \zz_{t-1}(i)^2) \bb_i(j)^2 \right] \\
& = \mathbb{E}_{t-1} \left[ \sum_{i \in A_{t}} \gamma_t^2 \yy_t(i)^2  \bb_i(j)^2 \right]
\tag{since $\mathbb{E}_{t-1}[\gamma_t] = 0$}
\end{align*}
By Claim \ref{clm:Dt}, 
\begin{align*}
\mathbb{E}_{t-1}[D_t - D_{t-1}] 
& \le (1+\eps) \mathbb{E}_{t-1} [\gamma_t^2 \cdot \tilde{\yy}_t^\top \DD_t \tilde{\yy}_t] \\
& = (1+\eps) \mathbb{E}_{t-1} \left[ \gamma_t^2 \cdot
\sum_{i \in A_t} \yy_t(i)^2 \sum_{j \in L_t} 
\lambda_j \bb_i(j)^2  \right] \\
& = (1+\eps) \mathbb{E}_{t-1} \left[ \gamma_t^2 \cdot
\sum_{j \in L_t} \lambda_j \sum_{i \in A_t} \yy_t(i)^2  \bb_i(j)^2 \right] \\
& = (1+\eps) \mathbb{E}_{t-1} \left[ \sum_{j \in L_t} \lambda_j (E_t(j) - E_{t-1}(j)) \right].
\end{align*}
Recall $\mathbb{E}_{t-1} [D_t - D_{t-1}] 
= \mathbb{E}_{t-1} \left[ \gamma_t^2 \left\langle \VV_{t,l}^\top \WW_t \VV_{t,l}, \tilde{\yy}_t \tilde{\yy}_t^\top \right\rangle \right]$ (the equation at the end of proof of Claim \ref{clm:Dt}), that is, only the rows indexed in $L_t$ contributes to the expected difference between $D_t$ and $D_{t-1}$.
Thus,
\begin{align*}
\mathbb{E}[D_T]
& \le (1+\eps) \mathbb{E} \left[ \sum_{j=1}^m  \lambda_j (E_{T}(j) - E_{t_j}(j))
\right] \\
& \le (1+\eps) \mathbb{E} \left[ \sum_{j=1}^m  \lambda_j E_{T}(j) \right] \\
& \le (1+\eps) (1+\frac{1}{\eps}) \sum_{j=1}^m \lambda_j \tag{since $\sum_{i\in A_{t_j}} \bb_i(j)^2 \le 1+\frac{1}{\eps}$} \\
& \le (1+\eps) (1+\frac{1}{\eps}).
\end{align*}
\end{proof}

\begin{proof}[Proof of Theorem \ref{thm:oracle_algo}]
Lemma \ref{lem:oracle_expect}, \ref{lem:tr_bound_uniform}, and \ref{lem:tr_bound_const} guarantees the expectation and covariance requirement on $\zz$. 
Since each while-iteration in Algorithm \ref{alg:oracle} turns at least one non $\pm 1$ entry of $\zz_{t-1}$ to $\pm 1$ and the algorithm only updates non $\pm 1$ entries, the total number of iterations is at most $n$. Each iteration runs in polynomial time, thus the total runtime is polynomial.
\end{proof}

\section{Statistical Characterizations of the MWU Design}
\label{sec:statistics}

In this section, we characterize the expected value, variance, consistency, and convergence rate of estimation error for the average treatment effect under the MWU design.
We prove Propositions \ref{prop:trade_off}, \ref{thm:unbias}, \ref{prop:var}, and \ref{prop:converge}.
The proofs follow those in \cite{HSSZ19}. We include them for completeness.

\begin{proof}[Proof of Proposition \ref{prop:trade_off}]
Let 
\[
\BB = \begin{pmatrix}
\sqrt{\phi} \II \\
\sqrt{1-\phi} \XX^\top
\end{pmatrix}
\]
be the augmented matrix defined as in Equation \eqref{eqn:mat_B}. We have $\norm{\BB}_{1,2} \le 1$.
By Theorem \ref{thm:algo}, $\norm{\cov(\BB \zz)} \le \alpha_{MWU}$. Therefore,
\[
\|\cov(\XX^\top \zz)\| \le \frac{\alpha_{MWU}}{1-\phi},
~ \|\cov(\zz)\| \le \frac{\alpha_{MWU}}{\phi}.
\]
\end{proof}

\begin{proof}[Proof of Proposition \ref{thm:unbias}]
The statement follows that $\zz$ returned by the MWU design is a feasible assignment and the definition of the HT estimator.
\end{proof}

\begin{proof}[Proof of Proposition \ref{prop:var}]
By Equation \eqref{eqn:mse_def}, 
\[
n \var(\hat{\tau}) = n \mathbb{E}[(\hat{\tau} - \tau)^2] 
= \frac{1}{n} \mmu^\top \cov(\zz) \mmu.
\]
By Theorem \ref{thm:algo} (rewriting the statement),
\[
\begin{pmatrix}
\sqrt{\phi} \II \\
\sqrt{1-\phi} \XX^\top
\end{pmatrix} \cov(\zz) 
\begin{pmatrix}
\sqrt{\phi} \II &
\sqrt{1-\phi} \XX
\end{pmatrix} \pleq \alpha_{MWU} \II,
\]
where $\pleq$ is the Loewner order for positive semidefinite matrices.
Since $\BB$ is full column rank, we have
\[
\cov(\zz) \pleq \alpha_{MWU} (\phi \II + (1-\phi) \XX \XX^\top)^{-1}.
\]
Thus,
\[
n \var(\hat{\tau}) \le \frac{\alpha_{MWU}}{n} \mmu^\top (\phi \II + (1-\phi) \XX \XX^\top)^{-1} \mmu.
\]
The right-hand side of the above equation is the optimal value of the ridge-regression given in the statement of Proposition \ref{prop:var}.
\end{proof}

\begin{proof}[Proof of Proposition \ref{prop:converge}]
By Proposition \ref{prop:var}, we have 
\begin{align*}
n \var(\hat{\tau}) & \le
\alpha_{MWU} \min_{\bbeta \in \mathbb{R}^d} \left\{ \frac{1}{\phi n} \|\mmu - \XX \bbeta \|^2 + \frac{1}{(1-\phi)n} \norm{\bbeta}^2 \right\} \\
& \le \frac{\alpha_{MWU}}{\phi n} \norm{\mmu}^2
\tag{substitute $\bbeta = {\bf 0}$}
\end{align*}
Note that 
\begin{align*}
\norm{\mmu}^2 & = \sum_{i=1}^n \left( \frac{a_i}{p_i} + \frac{b_i}{1-p_i} \right)^2 \\
& \le 2 \sum_{i=1}^n \frac{a_i^2}{p_i^2} + \frac{b_i^2}{(1-p_i)^2} \tag{by the Cauchy-Schwarz inequality} \\
& \le \frac{2}{c_2^2} \norm{\aa}^2 + \frac{2}{(1-c_2)^2} \norm{\bb}^2 
\tag{by Assumption (2)} \\
& \le c_4 n 
\tag{by Assumption (3)}
\end{align*}
So, 
\[
n \var(\hat{\tau}) \le \frac{c_4 \alpha}{c_1}.
\]
By Chebyshev's inequality, $\hat{\tau} - \tau \rightarrow 0$ in probability, and $\hat{\tau} - \tau = \mathcal{O}_p (n^{-1/2})$.
\end{proof}